\documentclass[11pt,a4paper]{article}

\usepackage{amsmath,amssymb,amscd}
\usepackage{amsthm,color,bm}
\usepackage{mathtools}
\usepackage{graphicx}
\usepackage{float}
\usepackage{tikz}
\usetikzlibrary{cd}
\usepackage{comment}
\usepackage{mathrsfs}
\usepackage{enumerate}
\usepackage{longtable}
\usepackage{array}
\usepackage{cases}
\usepackage{here}
\usepackage{authblk}
\numberwithin{equation}{section}
\allowdisplaybreaks[1]

\usepackage{mathabx} 

\usepackage{amsfonts}

\usepackage[top=3cm, bottom=3cm, left=2.5cm, right=2.5cm]{geometry}
\theoremstyle{plain}
\newtheorem{theorem}{Theorem}[section]
\newtheorem{proposition}[theorem]{Proposition}

\newtheorem{corollary}[theorem]{Corollary}

\theoremstyle{definition}
\newtheorem{definition}[theorem]{Definition}

\newtheorem{remark}[theorem]{Remark}
\newtheorem{notation}[theorem]{Notation}

\newcommand{\CC}{\mathbb{C}}
\newcommand{\RR}{\mathbb{R}}
\newcommand{\ZZ}{\mathbb{Z}}
\newcommand{\NN}{\mathbb{N}}
\newcommand{\PP}{\mathbb{P}}

\newcommand{\del}{\partial}
\newcommand{\dd}{\mathsf{d}}

\def\Set#1{\Setdef#1\Setdef}
\def\Setdef#1|#2\Setdef{\left\{#1\,\;\mathstrut\vrule\,\;#2\right\}}%

\begin{document}

\title{An elliptic fibration arising from the Lagrange top and its monodromy}
\author{Genki Ishikawa}
\affil{Department of Mathematical Sciences, Ritsumeikan University}
\date{}
\maketitle
\begin{abstract}
  This paper is to investigate an elliptic fibration over $\CC\PP^2$ arising from the Lagrange top from the viewpoint of complex algebraic geometry. The description of the discriminant locus of this elliptic fibration is given in detail. Moreover, the concrete description of the discriminant locus and the complete classification of singular fibres of the elliptic fibration are obtained according to Miranda's theory of elliptic threefolds after suitable modifications of the base and total spaces. Furthermore, the monodromy of the elliptic fibration is described.
\end{abstract}

\section{Introduction}

The heavy rigid body, i.e. a rigid body around a fixed point under gravity, is one of the typical solvable problems in analytical mechanics.
The rotational motion of a heavy rigid body can be formulated as a Hamiltonian system on the cotangent bundle to the three dimensional rotation group $SO(3)$ and because of the symmetry it can be reduced to a Hamiltonian system on the dual to the Lie algebra $\mathfrak{so}(3)\ltimes \mathbb{R}^3$.
Among such rigid bodies, it is known as a result of Ziglin \cite{ziglin_I-II} that the Hamiltonian systems are completely integrable exactly in the four particular cases: the Euler top, the Lagrange top, the Kowalevski top, and the Goryachev-Chaplygin top.

One of the basic results in the theory of completely integrable systems is Liouville-Arnol'd Theorem \cite{arnold_1989}. According to this theorem, if a generic common level manifold of first integrals is compact and connected, it is diffeomorphic to a torus and thus a completely integrable system gives rise to a torus fibration. Moreover, there are local charts called the action-angle coordinates on the torus fibration such that Hamilton equation is linearized on each torus. 
In general, global action-angle coordinates do not exist. Duistermaat \cite{duistermaat_1980} studied the obstructions to the existence of global action-angle coordinates including the monodromy.

Completely integrable systems have also been discussed from the viewpoint of complex algebraic geometry. One of the important concepts in the complex algebro-geometric theory of integrable systems is the algebraic complete integrability, which was introduced as the natural extension of Liouville integrability to complex algebraic category by Adler, van Moerbeke, and Mumford. See e.g. \cite{Adler-vanMoerbele-Vanhaecke,mumford_2007_II,vanhaecke_2001}.
An algebraic completely integrable system gives rise to a fibration whose generic fibres are abelian varieties and its flow can be linearized on each abelian variety.
For the systems described by Lax equations, the spectral curves and their Jacobian varieties have been studied to obtain the linearization of the flows for example in Adler and van Moerbeke \cite{Adler-vanMoerbeke2}, Gavrilov \cite{Gavrilov}, Beauville \cite{Beauville1990}, and Griffiths \cite{griffiths_1985}.
In particular, the Lagrange top has been studied by Gavrilov and Zhivkov \cite{Gavrilov-Zhivkov} and the spherical pendulum, which can be regarded as a special case of the Lagrange top, has been investigated by Beukers and Cushman \cite{Beukers-Cushman}, where the associated monodromy is also studied.

For algebraic completely integrable systems, the geometrical and dynamical properties of the original dynamical systems have mainly been investigated by focusing on each generic fibre of the associated fibrations. On the other hand,  few complex algebro-geometric studies are carried out in view of the singularities of the fibrations.

As for the Euler top, describing the integrable heavy rigid bodies without external force, complex algebro-geometric researches are made on the complex algebraic geometry of the fibrations by integral curves and by spectral curves by \cite{Naruki-Tarama,Tarama_Francoise_2014,Francoise-Tarama_2015}. 
They have connected the complex algebraic geometry of the associated elliptic fibrations, including the monodromy, to the Birkhoff normal forms or the action coordinates.

\medskip

The present paper deals with the complex algebraic geometry of an elliptic fibration induced by the complexified energy-momentum map for the Lagrange top. This elliptic fibration is given as an elliptic threefold over the complex projective plane $\CC\PP^2$ in Weierstra{\ss} normal form. 
The singular locus and complete classification of the singular fibres of the elliptic fibration are concretely described on the basis of Miranda's method discussed in \cite{Miranda1983}. 
Moreover, the monodromy of the elliptic fibration is studied.

The structure of the present paper is as follows: \\
In Section $2$, a review is given about a formulation of the completely integrable system for the Lagrange top with respect to the Lie-Poisson structure on the dual space $\left( \mathfrak{so}(3)\ltimes\RR^3 \right)^*$ of the semi-direct product Lie algebra. Each generic fibre of the energy-momentum map for the Lagrange top admits a free $S^1$-action and the quotient manifold of a generic fibre by the $S^1$-action is given as an affine part of the elliptic curve in Weierstra{\ss} normal form. 

Section $3$ is a brief overview about the classification of singular fibres of the elliptic fibration in Weierstra{\ss} normal form. Kodaira classified the singular fibres of minimal elliptic surfaces in \cite{KodairaI-III}. In the case of elliptic threefolds, Miranda gave the complete classification of singular fibres in \cite{Miranda1983} on the assumption that the reduced discriminant locus of elliptic threefolds permits only nodes as its singularities. In Miranda's classification, singular fibres sitting over generic points of each irreducible component of the discriminant locus belong to the list of Kodaira's classification, whereas singular fibres sitting over the nodes are not always in Kodaira's list. The list of these singular fibres is described at the end of this section.

In Section $4$, the elliptic threefold induced by the complexification of the energy-momentum map is investigated. 
The family of elliptic curves arising from the Lagrange top as in Section 2 can naturally be complexified and then compactified by a complex analytic elliptic fibration in Weierstra{\ss} normal form over $\mathbb{CP}^2$. 
The detailed description of the discriminant locus of this fibration, as well as the singularities of the total space, is given. 
After suitable modifications of base and total spaces, the singular fibres of the elliptic fibration are classified on the basis of Miranda's method \cite{Miranda1983} described in Section 3. 
Finally, the monodromy of the original elliptic threefold is described by Zariski-van Kampen Theorem \cite{Zariski_1929,vanKampen_1933}.

\section{The Lagrange top}

In this section, we review the Hamiltonian formalism of the Lagrange top and discuss its. 
The Lagrange top is a particular case of the heavy rigid bodies, i.e. the rigid bodies with a fixed point influenced by the gravity. 
The rotational motion of the heavy rigid body can be formulated as a Hamiltonian system on the dual $\left(\mathfrak{so}(3)\ltimes \mathbb{R}^3\right)^{\ast}$ of the semi-direct product Lie algebra with respect to the Lie-Poisson bracket. 
Here, we briefly describe this formalism and then characterize the Lagrange top as a special case. 
We then discuss its complete integrability when the system is restricted to generic coadjoint orbits.
See e.g. \cite{Ratiu-vanMoerbeke,audin_1999} for the details. 

Through the Lie algebra isomorphism $\left(\mathbb{R}^3, \times\right)\ni \begin{pmatrix}x_1 \\ x_2 \\ x_3\end{pmatrix}\mapsto\begin{pmatrix}0 & -x_3 & x_2 \\ x_3 & 0 & -x_1\\ -x_2 & x_1 & 0\end{pmatrix}\in\mathfrak{so}(3)$, we identify $\left(\mathfrak{so}(3)\ltimes \mathbb{R}^3\right)^{\ast}$ with $\mathbb{R}^3\times \mathbb{R}^3$, by using the standard inner product on $\mathbb{R}^3\times\mathbb{R}^3\cong\mathbb{R}^6$.

The Lie-Poisson structure on $\left(\mathfrak{so}(3)\ltimes \mathbb{R}^3\right)^{\ast}\cong \RR^3\times\RR^3$ is defined through
\begin{align*}
  \left\{F,G\right\}\left(\Gamma,M\right) 
    = -\langle\Gamma, \left( \nabla_MF\right)\times\left(\nabla_\Gamma G\right) \rangle -\langle\Gamma, \left(\nabla_\Gamma F\right)\times\left(\nabla_MG\right) \rangle-\langle M, \left( \nabla_MF\right)\times\left(\nabla_MG\right) \rangle,
\end{align*}
where $F,G\in\mathcal{C}^\infty\left( \RR^3\times\RR^3 \right)$. 
Here, $\left( \Gamma,M \right)\in\RR^3\times\RR^3$ and $\langle\cdot,\cdot\rangle$ stands for the standard inner product on $\RR^3$ and $\left( \nabla_\Gamma F,\nabla_M F \right)$ denotes the gradient of the function $F$ defined through
\begin{align*}
  \dd F_{\left( \Gamma,M \right)}\left( \xi,\eta \right)=\langle\nabla_\Gamma F\left( \Gamma,M \right),\xi\rangle +\langle\nabla_M F\left( \Gamma,M \right),\eta\rangle,
\end{align*}
where $\left( \xi,\eta \right)\in \RR^3\times\RR^3 \cong T_{\left( \Gamma,M \right)}\left( \RR^3\times\RR^3 \right)$. 

With respect to the Lie-Poisson bracket $\{\cdot,\cdot\}$, the heavy rigid body dynamics is formulated as a Hamiltonian system for the Hamiltonian given as 
\begin{align*}
H\left( \Gamma,M  \right)=\frac{1}{2}\langle M,\Omega\rangle + \langle\Gamma,\chi\rangle,\quad \Omega=\mathcal{J}^{-1}M, 
\end{align*}
where $\left( \Gamma,M \right)\in\RR^3\times\RR^3\cong\left( \mathfrak{so}(3)\ltimes\RR^3 \right)^*$. 
Physically, the constant positive-definite operator $\mathcal{J}: \mathbb{R}^3\rightarrow \mathbb{R}^3$ symmetric with respect to the standard inner metric is called the inertia matrix and the constant vector $\chi\in\mathbb{R}^3$ describes the position of the centre of mass relative to the fixed point. 
The corresponding Hamilton equation on the Poisson space $\left( \RR^3\times\RR^3,\{\cdot,\cdot\} \right)$ is described as
\begin{align}
  \begin{cases}
    \dot{\Gamma}=\Gamma\times \Omega,
    \\
    \dot{M}=M\times\Omega+\Gamma\times\chi,
  \end{cases}
  \label{Euler-Poisson eq}
\end{align}
which is often called as the Euler-Poisson equation.

As is well known, every Poisson manifold can be stratified into the disjoint union of symplectic leaves. 
See e.g. \cite{Marsden-Ratiu1999} for the details. 
For the Poisson space $\left( \RR^3\times\RR^3,\{\cdot,\cdot\} \right)$, the symplectic leaves coincide with the coadjoint orbits in $\left( \mathfrak{so}\left( 3 \right)\ltimes\RR^3 \right)^*$ equipped with the orbit symplectic forms. 
Those of the maximal dimension are described as 
\begin{align*}
  \mathcal{O}_a\coloneq\Set{\left( \Gamma,M \right)\in\RR^3\times\RR^3 | C_1\left( \Gamma,M \right)=1,\,C_2\left( \Gamma,M \right)=a},\,a\in\RR.
\end{align*}
Note that $\dim \mathcal{O}_a=4$. 
Here, $C_1$ and $C_2$ are defined through
\begin{align*}
  C_1\left( \Gamma,M \right)= \langle\Gamma,\Gamma\rangle,\,
  C_2\left( \Gamma,M \right)= \langle\Gamma,M\rangle.
\end{align*}
These are Casimir functions with respect to the Lie-Poisson bracket $\{\cdot,\cdot\}$, which by definition Poisson commute with every smooth function on $\RR^3\times\RR^3$. 

Recall that a Hamiltonian vector field on a Poisson manifold can be restricted to any symplectic leaf and the restricted system is again Hamiltonian with respect to the symplectic structure on the leaf. 
The restriction of the Euler-Poisson equations \eqref{Euler-Poisson eq} to $\mathcal{O}_a$ is therefore a Hamiltonian system on a four-dimensional symplectic manifold.

The heavy rigid body dynamics can also be formulated as a Hamiltonian system on $T^{\ast}SO(3)$ and its relation to the Lie-Poison formalism is explained in terms of the semi-direct product reduction theorem. See e.g. \cite{Ratiu_Tudoran_etc_2005} and \cite{Marsden-Ratiu-Weinstein} for the details.

The {\it{Lagrange top}} is the heavy symmetric top about the axis through the fixed point and the centre of mass can be characterized by $\mathcal{J}=\mathrm{diag}\left( J_1,J_2,J_3 \right)$ and $\chi=\left( \chi_1,\chi_2,\chi_3 \right)\in\RR^3$ satisfying
\begin{align*}
  J_1=J_2,\,\chi_1=\chi_2=0.
\end{align*}
On this assumption, this Hamiltonian system has the fourth constant of motion
\begin{align*}
  K\left( \Gamma,M \right)=-M_3, 
\end{align*}
in addition to the Hamiltonian $H$ and the two Casimir functions $C_1$, $C_2$. 
Here, $M=\left(M_1,M_2,M_3\right)^T$. 
Note that the four constants of motion $H$, $K$, $C_1$, $C_2$ are in involution and functionally independent. On the four-dimensional symplectic leaf $\mathcal{O}_a$, the restricted Hamiltonian system for the Hamiltonian $H$ admits another constant of motion $K$ and hence the system is completely integrable in the sense of Liouville. 

For the convenience of the later computations, we assume $\displaystyle\chi_3/J_1=-1$ by suitable reparametrizations and use the four constants of motion
\begin{align*}
  &H_1\coloneq C_1 = \Gamma_1^2+\Gamma_2^2+\Gamma_3^2,
  \\[5pt]
  &H_2\coloneq C_2/J_1 =\Gamma_1\Omega_1+\Gamma_2\Omega_2+\left( 1+m \right)\Gamma_3\Omega_3,
  \\[2pt]
  &H_3\coloneq H/J_1 = \frac{1}{2}\left( \Omega_1^2+\Omega_2^2+\left( 1+m \right)\Omega_3^2 \right)-\Gamma_3,
  \\[2pt]
  &H_4\coloneq -K/J_3 = \Omega_3,
\end{align*}
where $m=(J_3-J_2)/J_1$ and $\Gamma=\left( \Gamma_1,\Gamma_2,\Gamma_3 \right)^T,
  ~\Omega=\left( \Omega_1,\Omega_2,\Omega_3 \right)^T$, instead of $C_1,\,C_2,\,H,$ and $K$, following \cite{Gavrilov-Zhivkov}. 

Consider the energy-momentum mapping $\displaystyle \mathcal{E}\!\mathcal{M}=(H_3, H_4)\colon\mathcal{O}_a\to\RR^2$. 
A fibre of $\mathcal{E}\!\mathcal{M}$ can be written as
\begin{align*}
  \mathcal{E}\!\mathcal{M}^{-1}(h_3,h_4)=\Set{\left( \Gamma,\Omega \right)\in\RR^3\times\RR^3 | H_1\left( \Gamma,\Omega \right)=1,\,H_2\left( \Gamma,\Omega \right)=a,\,H_3\left( \Gamma,\Omega \right)=h_3,\,H_4\left( \Gamma,\Omega \right)=h_4},
\end{align*}
where $h_3$ and $h_4$ are suitable constants.
Since the constants of motion are polynomials in $\left(\Gamma, \Omega\right)$, it is an algebraic variety. 
This variety has the structure of an $S^1$-bundle over (the real part of) an elliptic curve as follows: \\
Consider the flow of the Hamiltonian vector field generated by $H_4$ whose Hamilton equation is written as
\begin{align*}
  \begin{cases}
    \dot\Omega_1=\Omega_2,
    \\
    \dot\Omega_2=-\Omega_1,
    \\
    \dot\Omega_3=0,
  \end{cases}
  \quad
  \begin{cases}
    \dot\Gamma_1=\Gamma_2,
    \\
    \dot\Gamma_2=-\Gamma_1,
    \\
    \dot\Gamma_3=0.
  \end{cases}
\end{align*}
The flow clearly induces an $S^1$-action on the generic fibres of the energy-momentum mapping $\mathcal{E}\!\mathcal{M}$ as
\begin{align*}
  S^1\times \mathcal{E}\!\mathcal{M}^{-1}(h_3,h_4)\ni\left( \theta,\left( \Gamma,\Omega \right) \right)\mapsto \left( A_{\theta}\Gamma,A_{\theta}\Omega  \right)\in \mathcal{E}\!\mathcal{M}^{-1}(h_3,h_4),
\quad   ~A_{\theta}=
  \begin{pmatrix}
    \cos\theta & -\sin\theta & 0 \\
    \sin\theta & \cos\theta & 0 \\
    0 & 0 & 1 \\
  \end{pmatrix}.
\end{align*}
Clearly, the polynomials $\Gamma_3$ and $\Gamma_1\Omega_2-\Gamma_2\Omega_1$ are invariant with respect to the $S^1$-action. 
Thus, the mapping 
\begin{align*}
  \mathcal{E}\!\mathcal{M}^{-1}(h_3,h_4)\ni\left( \Gamma,\Omega \right) \mapsto (x,y)=\left( -\frac{\Gamma_3}{2},-\frac{\Gamma_1\Omega_2-\Gamma_2\Omega_1}{2} \right)\in\RR^2,
\end{align*}
corresponds to the quotient mapping by the $S^1$-action. 
The image is the quotient space $$\mathcal{E}\!\mathcal{M}^{-1}(h_3,h_4)/S^1$$ and described as the real part of the complex cubic curve defined through 
\begin{align}
  y^2=4x^3-\left( 2h_3+\left( 1+m \right)mh_4^2 \right)x^2-\left( 1+\left( 1+m \right)ah_4 \right)x+\frac{1}{4}(2h_3-\left( 1+m \right)h_4^2-a^2), \label{cubic curve}
\end{align}
where $\left( x,y \right)\in\CC^2$.

\section{Elliptic fibrations}\label{Elliptic fibrations}

This section provides a brief overview of Miranda's theory of elliptic fibrations on complex analytic surfaces, based on \cite{Miranda1983}. 
In Section 4, we study a natural complex elliptic fibration over the complex projective plane in association to the Lagrange top, for which the basic notions are resumed in the present section. 

An elliptic fibration is a holomorphic 
surjection between two complex manifolds where the generic fibres are elliptic curves. 
A thorough fundamental study was carried out by Kodaira in \cite{KodairaI-III} in the case of elliptic surfaces, for which the base space is a compact complex curve and the total space is a compact complex surface. 
The singular fibres as well as the conjugacy classes of the monodromy matrices around them are classified. 
Further, the elliptic surfaces with a global holomorphic section is known to be described as in the Weierstra{\ss} normal form. 
See \cite{Kas1977} for the details. 

As a generalization of Kodaira's theory around elliptic surfaces, we here review Miranda's theory of elliptic threefolds in Weierstra{\ss} normal form, following \cite{Miranda1983}. 
In \cite{Miranda1983}, Miranda studied elliptic fibrations of complex three-dimensional manifolds to complex surfaces on which the singular loci are assumed to have only nodes as the singularities. 
As a result, the singular fibres for the elliptic threefolds are classified and it is shown that there appear singular fibres which cannot be observed in the case of elliptic surfaces. 
Some more recent developments can be found for example in \cite{Nakayama1988,Nakayama2002}.

\subsection{Weierstra{\ss} normal form}\label{Weierstrass normal form}

Let $S$ be a complex manifold and $\mathcal{L}$ a holomorphic line bundle over it. 
Take holomorphic sections $a\in H^0( S,\mathcal{L}^{\otimes4} )$ and $b\in H^0( S,\mathcal{L}^{\otimes6} )$ such that $\Delta\coloneq a^3-27b^2\in H^0\left(S, \mathcal{L}^{\otimes 12}\right)$ is not identically zero on $S$. 
We consider the vector bundle $\mathcal{L}^{\otimes2}\oplus\mathcal{L}^{\otimes3}\oplus\mathcal{O}_{S}$ of rank three over $S$ and denote its projectification by $P\left( \mathcal{L}^{\otimes2}\oplus\mathcal{L}^{\otimes3}\oplus\mathcal{O}_{S} \right):=(\mathcal{L}^{\otimes2}\oplus\mathcal{L}^{\otimes3}\oplus\mathcal{O}_{S})_0/\CC^*$, where the structure sheaf $\mathcal{O}_S$ of $S$ is identified with the trivial line bundle over $S$ and $(\mathcal{L}^{\otimes2}\oplus\mathcal{L}^{\otimes3}\oplus\mathcal{O}_{S})_0$ stands for the complement to the image of the zero section of the vector bundle $\mathcal{L}^{\otimes2}\oplus\mathcal{L}^{\otimes3}\oplus\mathcal{O}_{S}$ in the total space. 
We denote the natural projection of the projective bundle by $\pi\colon P\left( \mathcal{L}^{\otimes2}\oplus\mathcal{L}^{\otimes3}\oplus\mathcal{O}_{S} \right)\to S$. 

Let $\mathcal{W}$ be the divisor on the $\CC\PP^{2}$-bundle $P\left( \mathcal{L}^{\otimes2}\oplus\mathcal{L}^{\otimes3}\oplus\mathcal{O}_{S} \right)$ over $S$ defined through 
  \begin{align}
    Y^{2}Z -4X^{3}+aXZ^{2}+bZ^{3} =0, \label{WNF}
  \end{align}
where $\left(X:Y:Z\right)$ are the homogeneous fibre coordinates of $P\left( \mathcal{L}^{\otimes2}\oplus\mathcal{L}^{\otimes3}\oplus\mathcal{O}_{S} \right)$. 
Restricting the natural projection $\pi$ to $\mathcal{W}$, we obtain an elliptic fibration $\pi_{\mathcal{W}}\colon \mathcal{W}\to S$ over $S$.
  
\begin{definition}
The elliptic fibration $\pi_{\mathcal{W}}\colon \mathcal{W}\to S$ is said to be {\it in Weierstra{\ss} normal form} and the total space $\mathcal{W}$ is called the {\it{Weierstra{\ss} normal form}}. 
The section $\displaystyle \Delta= a^{3}-27b^{2}\in H^0\left(S, \mathcal{L}^{\otimes 12}\right)$ is called the {\it{discriminant}}, whereas the meromorphic function
  \begin{align*}
    J=\frac{a^3}{\Delta}=\frac{a^3}{a^3-27b^2},
  \end{align*}
on $S$ is called the {\it{functional invariant}} of the Weierstra{\ss} normal form.
\end{definition}

It is easily shown that the elliptic fibration $\pi_{\mathcal{W}}$ is flat, but its total space $\mathcal{W}$ may have singular points as described in the following proposition. 

\begin{proposition}{\rm{(\cite[Proposition 2.1]{Miranda1983})}}\label{prop:Sing}    
  Let $A$, $B$, and $D$ be the divisors on $S$ defined by $a=0$, $b=0$, and $\Delta=0$, respectively. 
  Moreover, let $\left(\left( X:Y:Z \right);p\right)$ denote a point of the total space $\mathcal{W}$ where $p\in S$. 
  Then, the following statements hold.
  \begin{enumerate}[$1.$]
  \item The total space $\mathcal{W}$ is smooth along $Z = 0$ and the set given by $\displaystyle\left(X:Y:Z\right)=\left(0:1:0\right)$ defines a holomorphic section of the elliptic fibration $\pi_\mathcal{W}$.
  \item If the total space $\mathcal{W}$ is singular at $\left(\left( X:Y:Z \right);p\right)$, then we necessarily have $Y=0$ and $Z\neq0$.
  \item The total space $\mathcal{W}$ is singular at $\left(\left( 0:0:Z \right);p\right)$ if and only if both $A$ and $B$ contain $p$ and $B$ is singular at $p$.	
  \item The total space $\mathcal{W}$ is singular at $\left(\left( X:0:Z \right);p\right)$ with $X\neq0$ if and only if neither $A$ nor $B$ contains $p$ but $D$ contains $p$ as a singular point of it. In this case, we have $\left(X:0:Z\right)=\left(-3b:0:2a\right)$.
  \end{enumerate}
\end{proposition}

From Proposition \ref{prop:Sing}, we see that the singular fibres of the elliptic fibration $\pi_\mathcal{W}$ lie over the points of the divisor $D$ on the base space $S$. Hence, the {\it singular locus}, or the {\it{discriminant locus}}, of the elliptic fibration $\pi_\mathcal{W}\colon\mathcal{W}\to S$ coincides with the divisor $D$. 

Note that, given a holomorphic mapping $f: S'\to S$ from a complex manifold $S^{\prime}$ to $S$, we can pull back the holomorphic line bundle $\mathcal{L}$ through $f$, as well as the holomorphic sections $a$, $b$, and $\Delta$. 
The equation $(Y^{\prime})^{2}(Z^{\prime}) -4(X^{\prime})^{3}+f^{\ast}aX^{\prime}(Z^{\prime})^{2}+f^{\ast}b(Z^{\prime})^{3} =0$ defines another elliptic fibration  $\pi_{\mathcal{W}'}\colon\mathcal{W}'\to S'$ in Weierstra{\ss} normal form over $S^{\prime}$. 
In what follows, this technique is often used to obtain a suitable form of an elliptic fibration through the modification of the base space within the same bimeromorphic class.

If the base space $S$ is a complex curve, the types of singular fibres of the elliptic fibration $\pi_\mathcal{W}\colon\mathcal{W}\to S$ are classified with the conjugacy classes of the monodromy matrices in \cite{KodairaI-III}. 
As described at Table \ref{Kosaira's-list} in the next section, these types of singular fibres appear also in the case of the elliptic threefolds.

\subsection{Miranda's elliptic threefolds}\label{Miranda's elliptic threefolds}

In this subsection, we consider an elliptic fibration $\pi_{\mathcal{W}}: \mathcal{W}\rightarrow S$ in Weierstra{\ss} normal form over a complex surface $S$. 
We denote the reduced divisors of $A$, $B$, and $D$ on $S$ by $A_0$, $B_0$, and $D_0$, respectively. 
Blowing up the base space $S$ if necessary, we may suppose that the following conditions:
\begin{enumerate}[(A)]
  \item The reduced discriminant locus $D_0$ permits only nodes as its singularities. \label{condition A}
  \item The divisors $A_0$ and $B_0$ have only nodes as their singularities. \label{condition B}
\end{enumerate}

If there exists a locally defined holomorphic function $u$ on $S$ satisfying $u^4|a$ and $u^6|b$, we may replace $X$ and $Y$ by $u^2X$ and $u^3Y$, respectively. 
Repeating this procedure if necessary, we also assume the following condition: 
\begin{enumerate}[(C)]
  \item If $u^4|a$ and $u^6|b$, then $u$ is a unit.
\end{enumerate}

On the assumptions (\ref{condition A})--(C), we may choose a local coordinate system $(s_1,s_2)$ on $S$ centred at $p\in \mathrm{supp}(D)$ in terms of which $a$, $b$, and $\Delta$ are written as
\begin{align}\label{germs_of_sections}
a(s_1,s_2)=s_1^{L_1}s_2^{L_2}\,\overline{a}(s_1,s_2),~b(s_1,s_2)=s_1^{K_1}s_2^{K_2}\,\overline{b}(s_1,s_2),~\Delta(s_1,s_2)=s_1^{N_1}s_2^{N_2}\,\overline{\Delta}(s_1,s_2),
\end{align}
respectively, where the germs of $\overline{a}$, $\overline{b}$, and $\overline{\Delta}$ at $p$ are units. 
As we are concerned with the types of singular fibres, we focus on an open neighbourhood of each point in the base space and hence the holomorphic functions are identified with their germs below, not being particularly mentioned.

\paragraph{Singular fibres over smooth points of $D_0$.}

If $p\in \mathrm{supp} (D)$ is a smooth point of $D_0$, we may assume $N_2=0$ around $p$. 
Then, since $\Delta=a^3-27b^2$ is not divisible by $s_2$, we necessarily have $L_2=0$ or $K_2=0$ and hence the functions $a$, $b$, and $\Delta$ can be written as
\begin{align*}
  a=s_1^{L1}\,\overline{a},~b=s_1^{K_1}\,\overline{b},~\Delta=s_1^{N_1}\,\overline{\Delta}. 
  \end{align*}

Let $T$ be a complex curve on the base space $S$ such that $T$ and the line $s_1=0$ intersect transversally. Restricting the elliptic fibration $\pi_{\mathcal{W}}: \mathcal{W}\rightarrow S$ to the inverse image of the curve $T$ through $\pi_{\mathcal{W}}$, we can obtain an elliptic surface over the curve $T$. Then, the type of the singular fibre over an intersection point of the curve $T$ and the line $s_1=0$ is determined in Kodaira's classification.

This means that to determine the types of singular fibres over the line $s_1=0$, it is enough to consider the case that $s_2$ is constant. Therefore, the singular fibres appearing over the points of $D$ in a neighbourhood of $p$ are of the Kodaira type determined by the triple $\left( L_1,K_1,N_1 \right)$.
Table \ref{Kosaira's-list} shows the list of these singular fibres corresponding to the triple $(L, K, N)$.
See \cite[{\S} 7.]{Miranda1989} for the details.
\begin{center}
  \begin{longtable}[c]{|m{0.1\textwidth}|m{0.09\textwidth}|m{0.18\textwidth}|m{0.3\textwidth}|}
  \caption{Kodaira's list}
  \label{Kosaira's-list}\\
  \hline
  Kodaira's notation & Dynkin diagram & $(L,K,N)$ types & types of singular fibres\\
  \hline
  \hline
  $I_0$ & -- & $(L\geq0,0,0)$ \par \hspace{2em}$\text{or}$ \par $(0,K\geq0,0)$  & Smooth elliptic curve  \\
  \hline
  $I_1$ & $A_0$ & $(0,0,1)$ &
  \begin{align*}
    \hspace{-1em}
    \underset{\text{Nodal rational curve}}{
    \begin{tikzpicture}[x=0.3cm, y=0.3cm] 
        \draw[line width=1pt] (0,0) .. controls (0,-1) and (1,-2) .. (2,-1);
        \draw[line width=1pt] (2,-1) --  (5,2);
        \draw[line width=1pt] (0,0) .. controls (0,1) and (1,2) .. (2,1);
        \draw[line width=1pt] (2,1) -- (5,-2);
    \end{tikzpicture} 
    }
    \end{align*}
    \\
  \hline
  $I_N$ & $A_{N-1}$ & $(0,0,N\geq 2)$ &
  \begin{align*}
    \hspace{-1em}
    \underset{\hspace{1em}\text{cycle of N smooth rational curve}}{
    \begin{tikzpicture}[x=0.4cm, y=0.4cm]
        \draw[line width=1pt] (-0.25,-0.25) -- (1,1.75);
        \draw[line width=1pt] (-0.25,0.25) -- (1,-1.75);
        \draw[line width=1pt] (0.5,1.5) -- (3,1.5);
        \draw[line width=1pt] (2.5,1.75) -- (3.75,-0.25);
        \draw[dashed] (3.75,0.25) -- (2.5,-1.75);
        \draw[dashed] (0.5,-1.5) -- (3,-1.5);
    \end{tikzpicture}
    }
  \end{align*}
  \\
  \hline
  $I_0^*$ & $D_4$ & $(L\geq2,K\geq3,6)$ &
  \begin{align*}
    \begin{tikzpicture}[x=0.6cm, y=0.6cm]
        \draw[line width=1pt] (-1.25,0) -- (1.25,0) node[right] {$2$};
        \draw[line width=1pt] (-0.75,0.7) -- (-0.75,-0.7) node[below] {$1$};
        \draw[line width=1pt] (-0.25,0.7) -- (-0.25,-0.7)node[below] {$1$};
        \draw[line width=1pt] (0.75,0.7) -- (0.75,-0.7)node[below] {$1$};
        \draw[line width=1pt] (0.25,0.7) -- (0.25,-0.7)node[below] {$1$};
    \end{tikzpicture}
  \end{align*}
  \\
  \hline
  $I_{N-6}^*$ & $D_{N-2}$ & $(2,3,N\geq7)$ &
  \begin{align*}
    \underbrace{
    \begin{tikzpicture}[x=0.8cm, y=0.8cm]
        \draw[line width=1pt] (0.25,0) -- (-0.75,1) node[midway,left] {$1$};
        \draw[line width=1pt] (0.5,0.25) -- (-0.5,1.25) node[midway,above] {$1$};
        \draw[line width=1pt] (0,0) -- (1,1) node[midway,right] {$2$};
        \draw[line width=1pt] (0.75,1) -- (1.75,0) node[midway,right] {$2$};
        \fill[black] (2,0.5) circle(0.05) (2.25,0.5) circle(0.05) (2.5,0.5) circle(0.05);
        \draw[line width=1pt] (2.75,0) -- (3.75,1) node[midway,left] {$2$};
        \draw[line width=1pt] (3.5,1) -- (4.5,0) node[midway,left] {$2$};
        \draw[line width=1pt] (4,0.25) -- (5,1.25) node[midway,above] {$1$};
        \draw[line width=1pt] (4.25,0) -- (5.25,1) node[midway,right] {$1$};
    \end{tikzpicture}
    }_{\text{$N-5$ multiplicity $2$ components}}
  \end{align*} 
  \\
  \hline
  $II$ & -- & $(L\geq1,1,2)$ & 
        \begin{align*}
          \hspace{-1em}
          \underset{\hspace{1em}\text{Cuspidal rational curve}}{
          \begin{tikzpicture}[x=0.4cm, y=0.4cm] 
              \draw[line width=1pt] (0,0) .. controls (1,0) and (2,1) .. (2,2); 
              \draw[line width=1pt] (0,0) ..controls (1,0) and (2,-1) .. (2,-2);
          \end{tikzpicture}
          }
        \end{align*}
    \\
  \hline
  $III$ & $A_1$ & $(1,K\geq2,3)$ &
  \begin{align*}
    \hspace{-0.5em}
    \begin{tikzpicture}[x=0.4cm, y=0.4cm]
        \draw[line width=1pt] (-1,1) .. controls (-0.5,1) and (0,0.5) ..  (0,0) ;
        \draw[line width=1pt] (-1,-1) node[left] {$1$} .. controls (-0.5,-1) and (0,-0.5) .. (0,0);
        \draw[line width=1pt] (1,1) .. controls (0.5,1) and (0,0.5) .. (0,0);
        \draw[line width=1pt] (1,-1) node[right] {$1$} .. controls (0.5,-1) and (0,-0.5) .. (0,0);
    \end{tikzpicture}
  \end{align*}
  \\
  \hline
  $IV$ & $A_2$ & $(L\geq2,2,4)$ &
  \begin{align*}
    \hspace{-1.5em}
    \begin{tikzpicture}[x=0.4cm, y=0.4cm]
        \draw[line width=1pt] (0,1) -- (0,-1) node[below] {$1$};
        \draw[line width=1pt] (-1,1) node[left] {$1$} -- (1,-1) ;
        \draw[line width=1pt] (-1,-1)node[left] {$1$} -- (1,1);
    \end{tikzpicture}
  \end{align*}
  \\
  \hline
  $IV^*$ & $E_6$ & $(L\geq3,4,8)$ &
  \begin{align*}
    \hspace{-0.5em}
    \begin{tikzpicture}[x=0.7cm, y=0.7cm]
        \draw[line width=1pt] (-1.25,0) -- (1.25,0);
        \draw[line width=1pt] (-1,0.25) -- (-1,-0.75)node[midway,left] {$2$};
        \draw[line width=1pt] (-1.75,-0.5) node[below] {$1$} -- (-0.75,-0.5)
        ;
        \draw[line width=1pt] (1,0.25) -- (1,-0.75) node[midway,right] {$2$};
        \draw[line width=1pt] (1.75,-0.5) node[below] {$1$}-- (0.75,-0.5);
        \draw[line width=1pt] (0,-0.25) node[right] {$3$} -- (0,1)node[midway, left] {$2$};
        \draw[line width=1pt] (-0.5,0.75) node[left] {$1$}   -- (0.5,0.75);
    \end{tikzpicture}
  \end{align*} 
  \\
  \hline
  $III^*$ & $E_7$ & $(3,K\geq5,9)$ &
  \begin{align*}
    \hspace{-0.5em}
    \begin{tikzpicture}[x=0.8cm, y=0.8cm]
        \draw[line width=1pt] (-1.15,0) -- (1.15,0);
        \draw[line width=1pt] (-1,0.15) -- (-1,-0.75)node[midway,left] {$3$};
        \draw[line width=1pt] (-1.75,-0.6)node[left] {$2$} -- (-0.85,-0.6);
        \draw[line width=1pt] (-1.6,-0.45) -- (-1.6,-1.35)node[left] {$1$};
        \draw[line width=1pt] (1,0.15) -- (1,-0.75)node[midway,right] {$3$};
        \draw[line width=1pt] (1.75,-0.6)node[right] {$2$} -- (0.85,-0.6);
        \draw[line width=1pt] (1.6,-0.45) -- (1.6,-1.35)node[right] {$1$};
        \draw[line width=1pt] (0,-0.25)node[right] {$4$} -- (0,0.65)node[left] {$2$};
    \end{tikzpicture}
  \end{align*}
  \\
  \hline
  $II^*$ & $E_8$ & $(L\geq4,5,10)$ &
  \begin{align*}
    \hspace{-2.5em}
    \begin{tikzpicture}[x=0.8cm, y=0.8cm]
        \draw[line width=1pt] (-1.15,0) -- (1.15,0);
        \draw[line width=1pt] (-1,0.15) -- (-1,-0.75)node[midway,left] {$5$};
        \draw[line width=1pt] (-1.75,-0.6)node[left] {$4$} -- (-0.85,-0.6);
        \draw[line width=1pt] (-1.6,-0.45) -- (-1.6,-1.35)node[midway,right] {$3$};
        \draw[line width=1pt] (-1.45,-1.2) -- (-2.35,-1.2)node[left] {$2$};
        \draw[line width=1pt] (-2.2,-1.05) -- (-2.2,-1.95)node[left] {$1$};
        \draw[line width=1pt] (1,0.15) -- (1,-0.75)node[midway,right] {$4$};
        \draw[line width=1pt] (1.75,-0.6)node[below] {$2$} -- (0.85,-0.6);
        \draw[line width=1pt] (0,-0.25)node[right] {$6$} -- (0,0.65)node[left] {$3$};
    \end{tikzpicture}
  \end{align*}
  \\
  \hline
  \end{longtable}
\end{center}

\paragraph{Singular fibres over singular points of $D_0$.}
Assume that $D_0$ is singular at $p$ and take a local coordinate $(s_1,s_2)$ centred at $p$.
As $(s_1,s_2)=(0,0)$ is a singular point of $D_0$, we obtain $N_1>0,\, N_2>0$ in the description \eqref{germs_of_sections}.
In this case, the reduced discriminant locus $D_0$ is defined through the equation $s_1s_2=0$ around $p$.

Now the total space $\mathcal{W}$ of the elliptic fibration is locally described through the equation
\begin{equation}\label{modified-WNF}
  Y^2Z = 4X^3-s_1^{L_1}s_2^{L_2}\,\overline{a}\,XZ^2-s_1^{K_1}s_2^{K_2}\,\overline{b}\,Z^3,
\end{equation}
whose discriminant is written as 
\begin{align}
  \Delta=s_1^{N_1}s_2^{N_2}\,\overline{\Delta}. \label{modified-discriminant}
\end{align}
From the above argument, the types of singular fibres over the lines $s_1=0$ and $s_2=0$ except for the node $s_1=s_2=0$ are determined by the triples $\left( L_1,K_1,N_1 \right)$ and $\left( L_2,K_2,N_2 \right)$, respectively.

\medskip
We consider a blowing-up $\sigma\colon S_1\to S$ of $S$ with the centre at $p$ and the pull-back $\pi_1\colon\mathcal{W}_1\to S_1$ of the fibration $\pi\colon\mathcal{W}\to S$ through $\sigma$. 
The singular fibres of $\pi_1$ over the exceptional divisor $E$ of $\sigma$ are described as in the following proposition. For the proof, see \cite{Miranda1983}.

\begin{proposition}{\rm{\cite[Proposition 9.1]{Miranda1983}}}\label{prop:types-over-E}
  Assume that the defining equation of $\mathcal{W}$ and the discriminant around p are written as in the forms {\rm{(\ref{modified-WNF})}} and {\rm{(\ref{modified-discriminant})}}. 
  Then, after a change of the local base and fibre coordinates of $\pi_1$ if necessary, we can assume that the condition {\rm{(C)}} is satisfied. 
  As a result, the types of the singular fibres of $\pi_1$ over the exceptional divisor $E$ in $S_1$ are determined by the triple $(L_1+L_2,K_1+K_2,N_1+N_2)$ modulo $(4,6,12)$.
\end{proposition}

In \cite{Miranda1983}, after suitable blowing-ups of $S$ and a resolution $\mathcal{\widehat{\mathcal{W}}}\to\mathcal{W}$ of singularities of the total space, a classification is carried out for the possible collisions of Kodaira's fibre types between the irreducible components $s_1=0$ and $s_2=0$ of $D_0$ in a neighbourhood of $p$. 
The singular fibres of the elliptic fibration $\pi_{\widehat{\mathcal{W}}}\colon\widehat{\mathcal{W}}\to S$ over the collision points are described for each colliding type as in Table \ref{Miranda's list}.

\begin{remark}
  It should be emphasized that, according to Miranda's classification, the singular fibres over the double points of $D_0$ are obtained by contracting some of components of the singular fibres of Kodaira's types as is indicated in the last column of Table \ref{Miranda's list}.
\end{remark}

\newpage
\begin{center}
  \begin{longtable}[c]{|m{0.14\textwidth}|m{0.18\textwidth}|m{0.16\textwidth}|m{0.16\textwidth}|}
  \caption{List of Miranda's singular fibres}
  \label{Miranda's list}\\
  \hline
  colliding types & \hspace{-13ex}fibre of $\pi_{\widehat{\mathcal{W}}}$ over $p$ & corresponding\par Kodaira's types & contracted\par components\\
  \hline
  \hline
  $I_{M_1}+I_{M_2}$ 
  &
  \begin{align*}
  \hspace{-3em}
    \underset{\hspace{1em}\text{cycle of $(M_1+M_2)$ smooth  rational curves}}{
    \begin{tikzpicture}[x=0.35cm, y=0.35cm]
        \draw[line width=1pt] (-0.25,-0.25) -- (1,1.75);
        \draw[line width=1pt] (-0.25,0.25) -- (1,-1.75);
        \draw[line width=1pt] (0.5,1.5) -- (3,1.5);
        \draw[line width=1pt] (2.5,1.75) -- (3.75,-0.25);
        \draw[dashed] (3.75,0.25) -- (2.5,-1.75);
        \draw[dashed] (0.5,-1.5) -- (3,-1.5);
    \end{tikzpicture}
    }
  \end{align*}
  & $I_{M_1+M_2}$ & none  \\
  \hline
  $I_{M_1}+I_{M_2}^*$ \par \vspace{0.5em} ($M_1$: even) 
  &
  \begin{align*}
    \underbrace{
    \begin{tikzpicture}[x=0.8cm, y=0.8cm]
        \draw[line width=1pt] (0.25,0) -- (-0.75,1) node[midway,left] {$1$};
        \draw[line width=1pt] (0.5,0.25) -- (-0.5,1.25) node[midway,above] {$1$};
        \draw[line width=1pt] (0,0) -- (1,1) node[midway,right] {$2$};
        \draw[line width=1pt] (0.75,1) -- (1.75,0) node[midway,right] {$2$};
        \fill[black] (2,0.5) circle(0.05) (2.25,0.5) circle(0.05) (2.5,0.5) circle(0.05);
        \draw[line width=1pt] (2.75,0) -- (3.75,1) node[midway,left] {$2$};
        \draw[line width=1pt] (3.5,1) -- (4.5,0) node[midway,left] {$2$};
        \draw[line width=1pt] (4,0.25) -- (5,1.25) node[midway,above] {$1$};
        \draw[line width=1pt] (4.25,0) -- (5.25,1) node[midway,right] {$1$};
    \end{tikzpicture}
    }_{\text{$\left( M_2+\frac{M_1}{2}+1 \right)$ components with multiplicity $2$}}
  \end{align*} 
  & $I_{M_1+M_2}^*$ & {$\frac{M_1}{2}$ components with multiplicity $2$} 
    \\
  \hline
  $I_{M_1}+I_{M_2}^*$ \par \vspace{0.5em}($M_1$: odd) 
  & 
  \begin{align*}
    \underbrace{
    \begin{tikzpicture}[x=0.9cm, y=0.9cm]
        \draw[line width=1pt] (0.25,0) -- (-0.75,1) node[midway,left] {$1$};
        \draw[line width=1pt] (0.5,0.25) -- (-0.5,1.25) node[midway,above] {$1$};
        \draw[line width=1pt] (0,0) -- (1,1) node[midway,right] {$2$};
        \draw[line width=1pt] (0.75,1) -- (1.75,0) node[midway,right] {$2$};
        \fill[black] (2,0.5) circle(0.05) (2.25,0.5) circle(0.05) (2.5,0.5) circle(0.05);
        \draw[line width=1pt] (2.75,0) -- (3.75,1) node[midway,left] {$2$};
        \draw[line width=1pt] (3.5,1) -- (4.5,0) node[midway,left] {$2$};
    \end{tikzpicture}
    }_{\text{$\left( M_2+\frac{M_1-1}{2}+1 \right)$ components with multiplicity $2$}}
  \end{align*}
  & $I_{M_1+M_2}^*$ & {$\frac{M_1-1}{2}$ components with multiplicity $2$ and $2$ components with multiplicity $1$}
  \\
  \hline
  $II+IV$ 
  &
  \begin{align*}
   \hspace{-8em}
    \begin{tikzpicture}[x=0.9cm, y=0.9cm]
        \draw[line width=1pt] (-1.25,0) -- (1.25,0) node[right] {$2$};
        \draw[line width=1pt] (-0.75,0.7) -- (-0.75,-0.7) node[below] {$1$};
    \end{tikzpicture}
  \end{align*}
  & $I_0^*$ & $3$ {components with multiplicity 1}
  \\
  \hline
  $II+I_0^*$ 
  & 
  \begin{align*}
  \hspace{-9.5em}
    \begin{tikzpicture}
        \draw[line width=1pt] (-1.25,0)node[left] {$3$} -- (1.25,0);
        \draw[line width=1pt] (0,-0.25)  -- (0,1)node[midway, left] {$2$};
        \draw[line width=1pt] (-0.5,0.75) node[left] {$1$}   -- (0.5,0.75);
    \end{tikzpicture}
  \end{align*}
  & 
  $IV^*$ 
  &two of the three
  \begin{align*}
    \begin{tikzpicture}[x=0.4cm, y=0.4cm]
        \draw[line width=1pt] (-1.2,0) -- (1.2,0) node[right] {$1$};
        \draw[line width=1pt] (0,-1) -- (0,1) node[above] {$2$};
    \end{tikzpicture}
  \end{align*}
  components
  \\
  \hline
  $II+IV^*$ 
  & 
  \begin{align*}
  \hspace{-9em}
    \begin{tikzpicture}
        \draw[line width=1pt] (-1.15,-0.55) -- (-1.7,-1.35)node[midway,left] {$3$};
        \draw[line width=1pt] (-1.45,-1.2) -- (-2.35,-1.2)node[left] {$2$};
        \draw[line width=1pt] (-2.2,-1.05) -- (-2.2,-1.95)node[left] {$1$};
        \draw[line width=1pt] (-1.3,-0.55) -- (-0.7,-1.35)node[midway,right] {$4$};
        \draw[line width=1pt] (-0.1,-1.2)node[below] {$2$} -- (-0.95,-1.2);
    \end{tikzpicture}
  \end{align*}
  & 
  $II^*$ 
  & 
  \begin{align*}
    \begin{tikzpicture}
        \draw[line width=1pt] (-1.15,0) -- (0,0) node[right] {$6$};
        \draw[line width=1pt] (-1,0.15) -- (-1,-0.75)node[midway,left] {$5$};
        \draw[line width=1pt] (-1.75,-0.6)node[left] {$4$} -- (-0.85,-0.6);
        \draw[line width=1pt] (-0.25,-0.25) -- (-0.25,0.65)node[left] {$3$};
    \end{tikzpicture}
  \end{align*}
  components
        
    \\
  \hline
  $IV+I_0^*$ 
  & 
  \begin{align*}
  \hspace{-8em}
    \begin{tikzpicture}
        \draw[line width=1pt] (-1.4,-1.2) -- (-2.35,-1.2)node[midway,above] {$2$};
        \draw[line width=1pt] (-2.2,-1.05) -- (-2.2,-1.95)node[midway,left] {$1$};
        \draw[line width=1pt] (-1.55,-1.05) -- (-1.55,-1.95)node[midway,right] {$4$};
        \draw[line width=1pt] (-1.7,-1.8) -- (-0.8,-1.8) node[right] {$2$};
    \end{tikzpicture}
  \end{align*}
  & 
  $II^*$ 
  &
  \begin{align*}
    \begin{tikzpicture}
        \draw[line width=1pt] (-1.15,0) -- (0,0) node[right] {$6$};
        \draw[line width=1pt] (-1,0.15) -- (-1,-0.75)node[midway,left] {$5$};
        \draw[line width=1pt] (-1.75,-0.6)node[left] {$4$} -- (-0.85,-0.6);
        \draw[line width=1pt] (-1.6,-0.45) -- (-1.6,-1.35)node[right] {$3$};
        \draw[line width=1pt] (-0.25,-0.25) -- (-0.25,0.65)node[left] {$3$};
    \end{tikzpicture}
  \end{align*}
  components
  \\
  \hline
  $III+I_0^*$ 
  & 
  \begin{align*}
  \hspace{-8em}
    \begin{tikzpicture}
        \draw[line width=1pt] (0.25,0) -- (1.15,0) node[midway,above] {$2$};
        \draw[line width=1pt] (0.4,0.15) -- (0.4,-0.75)node[left] {$1$};
        \draw[line width=1pt] (1,0.15) -- (1,-0.75)node[midway,right] {$3$};
        \draw[line width=1pt] (1.75,-0.6)node[right] {$2$} -- (0.85,-0.6);
        \draw[line width=1pt] (1.6,-0.45) -- (1.6,-1.35)node[right] {$1$};
    \end{tikzpicture}
  \end{align*}
  & 
  $III^*$ 
  &
  \begin{align*}
    \begin{tikzpicture}
        \draw[line width=1pt] (-1.15,0) -- (0.15,0)node[right] {$4$};
        \draw[line width=1pt] (-1,0.15) -- (-1,-0.75)node[midway,left] {$3$};
        \draw[line width=1pt] (-0.25,-0.25) -- (-0.25,0.65)node[left] {$2$};
    \end{tikzpicture}
  \end{align*}
  components
  \\
  \hline
  \end{longtable}
\end{center}

\section{The elliptic fibration associated with the Lagrange top}

In this section, we consider an elliptic fibration in Weierstra{\ss} normal form naturally constructed in association to the Lagrange top. 
The geometry of the elliptic fibration is thoroughly studied from the viewpoint of Miranda's elliptic threefolds. 
In view of this complex algebro-geometric framework, all the settings in Section 2 are now complexified in this section. 
In Subsection 4.1, we consider a family of complex plane cubic curves arising from the Lagrange top and then compactify it as an elliptic fibration over $\mathbb{CP}^2$, denoted as $\pi_W: W\rightarrow \mathbb{CP}^2$. 
This elliptic threefold has singularities as discussed concretely in Subsection 4.2. 
In Subsection 4.3, we find the smooth elliptic fibration $\pi_{\widehat{\mathcal{W}}}: \widehat{\mathcal{W}}\rightarrow \widehat{\mathbb{CP}^2}$ which is bimeromorphic to the singular elliptic fibration and satisfies the conditions in Miranda's study on elliptic threefolds as in the last section. 
It requires the modifications of the total space and the base plane of the elliptic fibration. 
Subsection 4.4 deals with the monodromy of the elliptic fibration. 
The elliptic fibrations appearing in their sections are related to each other in Figure 1. 

\begin{figure}[H]
  \centering
  \caption{elliptic fibrations in Section 4.}
  \begin{tikzcd}
    \mathcal{O}_a \arrow[d,"\mathcal{E}\!\mathcal{M}^{\CC}"]\arrow[r,hookrightarrow] & W \arrow[r,dashed]\arrow[d,"\pi_W"] & \widehat{\mathcal{W}} \arrow[d,"\pi_{\widehat{\mathcal{W}}}"] \\
    \CC^2 \arrow[r,hookrightarrow] & \CC\PP^2 \arrow[r,dashed] & \widehat{\CC\PP^2}
  \end{tikzcd}
\end{figure}

\subsection{A family of complex elliptic curves and its compactification}\label{Formulation of W}
We complexify all the settings in Section 2 for the convenience of the subsequent analysis. 
Namely, the four polynomials $H_1$, $H_2$, $H_3$, $H_4$ are naturally extended to those on $\mathbb{C}^3\times\mathbb{C}^3$, where the parameters $m$ and $a$ are also regarded as complex numbers. 
The Lie algebraic and the Poisson structures are also holomorphically extended to $\mathbb{C}^3\times\mathbb{C}^3$ and, as a result, the generic coadjoint orbits $\mathcal{O}_a^{\mathbb{C}}$ are regarded as complex affine algebraic varieties of dimension four, which can be characterized as the intersections of two quadric hypersurfaces $H_1(\Gamma, \Omega)=1$, $H_2(\Gamma, \Omega)=a$, where $a\in \mathbb{C}$. 
The energy-momentum mapping is also complexified and denoted as $\mathcal{E}\!\mathcal{M}^{\CC}: \mathcal{O}_a^{\mathbb{C}}\ni (\Gamma, \Omega)\mapsto (H_1(\Gamma, \Omega), H_2(\Gamma, \Omega))\in \mathbb{C}^2$. 
The holomorphic Hamiltonian flow associated to $H_4$ induces a free action of $\mathbb{C}^{\ast}\cong \mathbb{C}/2\pi \sqrt{-1}$ on each regular fibre of $\mathcal{E}\!\mathcal{M}^{\CC}$. 
See \cite[Appendix]{Gavrilov-Zhivkov} for more details of the description of the $\mathbb{C}^{\ast}$-action. 
The quotient curve of this action can still be written as \eqref{cubic curve}, which defines a complex elliptic curve in $\mathbb{C}^2$ with $h_3$, $h_4$ being complex parameters. 
For the convenience of the later calculations, we introduce the transformation of the parameters $\tau: \mathbb{C}^2\ni (h_3, h_4)\mapsto (a_1, a_2)\in \mathbb{C}^2$ defined by 
\begin{align*}
  a_1=2\left( 1+m \right)h_4,~a_2=2h_3+\left( 1+m \right)mh_4^2. 
\end{align*}
Setting $\alpha=-2a$ and replacing $x-\displaystyle \dfrac{a_2}{12}$ by $x$, we obtain an affine cubic curve $$\mathcal{C}_{(a_1,a_2)}\cong\left(\mathcal{E}\!\mathcal{M}^{\CC}\right)^{-1}(h_3,h_4)/\CC^*$$ in Weierstra{\ss} normal form as
\begin{align}
y^2=4x^3-g_2x-g_3,\label{eq_Weierstrass}
\end{align}
where
\begin{align}
  g_2=1+\frac{a_2^2}{12}-\frac{\alpha}{4}a_1,~g_3=\frac{a_2^3}{216}+\frac{a_1^2}{16}-\frac{\alpha}{48}a_1a_2-\frac{1}{6}a_2+\frac{\alpha^2}{16}.\label{g_2 g_3}
\end{align}
The family of the elliptic curve $\mathcal{C}_{(a_1,a_2)}$ with the parameter $(a_1, a_2)\in \mathbb{C}^2$ can be compactified by an elliptic fibration in Weierstra{\ss} normal form. 

We compactify the family $\left\{\mathcal{C}_{(a_1,a_2)}\right\}_{(a_1, a_2)\in \mathbb{C}^2}$ of plane curves by an elliptic fibration in Weierstra{\ss} normal form over $\mathbb{CP}^2$. 
Consider the hyperplane line bundle $\mathcal{L}:=\mathcal{O}_{\mathbb{CP}^2}(1)$ with the first Chern number $1$ over the complex projective plane $\mathbb{CP}^2$ whose homogeneous coordinates are $\left(A_0: A_1: A_2\right)$. 
On each Zariski open subset $U_i:=\left\{(A_0: A_1: A_2)\in\mathbb{CP}^2\mid A_i\neq 0\right\}$, $i=0,1,2$, we consider the  holomorphic functions 
\begin{align}
\left(g_2^{\ast}\right)_i = \frac{1}{A_i^4}\Phi\left(A_0,A_1,A_2\right),\;\left(g_3^{\ast}\right)_i = \frac{1}{A_i^6}\Psi\left(A_0,A_1,A_2\right),
\end{align}
where 
\begin{align}
&\Phi\left(A_0,A_1,A_2\right)=A_0^2\left(A_0^2 +\frac{1}{12}A_2^2 - \frac{\alpha}{4}A_0A_1\right), \label{Phi}
\\[4pt]
&\Psi\left(A_0,A_1,A_2\right)=A_0^3\left(\frac{1}{216}A_2^3 + \frac{1}{16}A_0A_1^2 -\frac{\alpha}{48}A_0A_1A_2 -\frac{1}{6}A_0^2A_2 +\frac{\alpha^2}{16}A_0^3\right).\label{Psi}
\end{align}
Clearly, the collections $g_2^{\ast}=\left(\left(g_2^{\ast}\right)_i\right)_{i=0,1,2}$ and $g_3^{\ast}=\left(\left(g_3^{\ast}\right)_i\right)_{i=0,1,2}$ are holomorphic sections of $\mathcal{L}^{\otimes 4}$ and $\mathcal{L}^{\otimes 6}$ over $\mathbb{CP}^2$, respectively. 
Then, using the general framework of the elliptic fibrations in Weierstra{\ss} normal form in Subsection 3.1, we define the Weierstra{\ss} normal form $W\subset P\left(\mathcal{L}^{\otimes2}\oplus\mathcal{L}^{\otimes3}\oplus\mathcal{O}_{\CC\PP^2}\right)$ defined through 

\begin{align}
Y^2Z = 4 X^3 -g_2^* X Z^2 -g_3^* Z^3. \label{eq_WM}
\end{align}
As in the previous section, $(X:Y:Z)$ are the homogeneous fibre coordinates of the $\mathbb{CP}^2$-bundle $P\left(\mathcal{L}^{\otimes2}\oplus\mathcal{L}^{\otimes3}\oplus\mathcal{O}_{\CC\PP^2}\right)$. 
The restriction $\pi_W\colon W\to\CC\PP^2$ of the canonical projection $\pi\colon P\left(\mathcal{L}^{\otimes2}\oplus\mathcal{L}^{\otimes3}\oplus\mathcal{O}_{\CC\PP^2}\right)\to\CC\PP^2$ to $W$ is an elliptic fibration in Weierstra{\ss} normal form. 
On the base plane $\mathbb{CP}^2$, the divisors defined through the equation $g_2^{\ast}=0$ and $g_3^{\ast}=0$ are denoted by $G_2$ and $G_3$, respectively. 
We also consider the discriminant $\Delta=\left(g_2^*\right)^3-27\left(g_3^*\right)^2\in H^0\left(\mathbb{CP}^2, \mathcal{L}^{\otimes 12}\right)$ and the functional invariant $J=\frac{\left(g_2^*\right)^3}{\Delta}=\frac{\left(g_2^*\right)^3}{\left(g_2^*\right)^3-27\left(g_3^*\right)^2}$, which is a rational function on $\mathbb{CP}^2$. 
Note that $\left(g_2^{\ast}\right)_0=g_2$ and $\left(g_2^{\ast}\right)_0=g_3$ in \eqref{g_2 g_3} and that we find again the equation \eqref{eq_Weierstrass}, setting $(x,y)=(X/Z, Y/Z)$ and $(a_1, a_2)=(A_1/A_0, A_2/A_0)$. 
Thus, the elliptic fibration $\pi_W: W\rightarrow \mathbb{CP}^2$ is regarded as the compactification of the family $\left\{\mathcal{C}_{(a_1,a_2)}\right\}_{(a_1, a_2)\in \mathbb{C}^2}$. 

\subsection{Singularities of $W$ and singular locus of $\pi_W$}

In this subsection, we determine the singularities of the threefold $W$ and analyse the singular locus of the elliptic fibration $\pi_W\colon W\rightarrow \mathbb{CP}^2$ defined as in \eqref{eq_WM}. 

\begin{proposition}
  The set of singularities of $W$ is given as
  \begin{align*}
    {\rm{Sing}}(W) &= \left\{\left( \left( -\frac{\alpha^2}{48}\pm\frac{1}{3}:0:1 \right);\left( 1:\pm\alpha:\frac{\alpha^2}{4}\pm2 \right) \right)\in W\right\} \notag
    \\[4pt]
    &\hspace{15em}\cup\Set{\left( \left( 0:0:1 \right);\left( 0:A_1:A_2 \right) \right)\in W| \left(A_1: A_2\right)\in\CC\PP^1},
  \end{align*}
  where $\left( \left( X:Y:Z \right);\left( A_0:A_1:A_2 \right) \right)$ denotes a point of $P\left(\mathcal{L}^{\otimes2}\oplus\mathcal{L}^{\otimes3}\oplus\mathcal{O}_{\CC\PP^2}\right)$.
\end{proposition}
\begin{proof}
  From Proposition \ref{prop:Sing}, it is enough to consider the open set where $Z\neq0$. 
  Over the Zariski open set $U_0 = \{A_0\neq0\}$ in $\CC\PP^2$, the equation \eqref{eq_WM} is written as \eqref{eq_Weierstrass} where $(x,y)=\left( X/Z,Y/Z \right)$ and $(a_1,a_2)=(A_1/A_0,A_2/A_1)$. 
  To find the singular points, we compute the derivative of the function
  \begin{align*}
    \varphi\left( x,y,a_1,a_2 \right) \coloneq y^2 - 4x^3+\left( 1 + \frac{a_2^2}{12}-\frac{\alpha}{4}a_1 \right) x + \left( \frac{a_2^3}{216}+\frac{a_1^2}{16}-\frac{\alpha}{48}a_1a_2-\frac{1}{6}a_2+\frac{\alpha^2}{16} \right).
  \end{align*}
  At a singular point of $W$, we have $\displaystyle \left(\frac{\del \varphi}{\del x}, \frac{\del \varphi}{\del y}, \frac{\del \varphi}{\del a_1}, \frac{\del \varphi}{\del a_2} \right) = \left( 0,0,0,0 \right)$, which is equivalent to
  \begin{numcases}{}
    1-\frac{\alpha}{4}a_1 = \frac{1}{12}\left( 12x-a_2 \right)\left( 12x+a_2 \right), \label{par-x}
    \\[4pt]
    y=0, \label{par-y}
    \\[4pt]
    a_1 = \frac{\alpha}{6}\left( 12x+a_2 \right), \label{par-a_1}
    \\[2pt]
    \frac{a_2}{12}\left( 12x+a_2 \right)=\frac{\alpha}{8}a_1+1 \label{par-a_2}.
  \end{numcases}
  By summing up $\text{\eqref{par-a_1}}^2$ and $\dfrac{2\alpha^2}{3}\times \text{\eqref{par-a_2}}$, we have $a_1=\pm\alpha$, using \eqref{par-x}.
  Henceforth, the double-sings correspond to each other.  
  Then, substituting \eqref{par-a_1} and \eqref{par-a_2} by $a_1=\pm\alpha$, we have $a_2=\dfrac{\alpha^2}{4}\pm 2$. 
  Recall that $\alpha$ is generically chosen and hence we assume $\alpha\neq 0$. 
  We further have $x=-\dfrac{\alpha^2}{48}\pm\dfrac{1}{3}$, using again \eqref{par-a_1}. 
  Thus, by \eqref{par-y}, the two points 

  \begin{align*}
    \left( \left( -\frac{\alpha^2}{48}\pm\frac{1}{3}:0:1 \right);\left( 1:\pm\alpha:\frac{\alpha^2}{4}\pm2 \right) \right),
  \end{align*}
  are singular points of $W$.

  Over $U_1=\{A_1\neq0\}$, the equation (\ref{eq_WM}) is written as in the following form:
  \begin{align*}
    y^2=4x^3-\left(g_2^{\ast}\right)_1(u,v)x-\left(g_3^{\ast}\right)_1(u,v),
  \end{align*}
  where $(x,y)=(X/Z,Y/Z)$ and $(u,v)=(A_0/A_1,A_2/A_1)$. 
  As $\left(g_2^{\ast}\right)_1$ and $\left(g_3^{\ast}\right)_1$ are respectively divisible by $u^2$ and $u^3$, we see that the hypersurface defined through this equation is singular along the curve $(x,y)=(0,0)$, $u=0$. 
  Thus, the set 
  \[
\Set{\left( \left( 0:0:1 \right);\left( 0:1:v \right) \right)\in W|v\in\CC} 
  \] is contained in ${\rm{Sing}}(W)$. Similarly, we can verify that the point $\left( \left( 0:0:1 \right);\left( 0:0:1 \right) \right)\in W$ is a singular point of $W$.
\end{proof}

By suitable blowing-ups of $W$ along the singular set $\mathrm{Sing}(W)$, we obtain a smooth model $\widehat{W}\to \CC\PP^2$ of the elliptic fibration. However, this elliptic fibration does not satisfy the condition (\ref{condition A}) and (\ref{condition B}). In fact, the divisors $G_2$ and $G_3$ on $\CC\PP^2$ defined respectively by $g_2^*=0$ and $g_3^*=0$ have worse singularities than nodes as observed in the following proposition. 

\begin{proposition}\label{prop:intersection}
  Let $\widetilde{G_2}$ and $\widetilde{G_3}$ be the divisors on $\CC\PP^2$ defined by the homogeneous polynomials $\widetilde{\Phi}=0$ and $\widetilde{\Psi}=0$, respectively, where 
  \begin{align*}
    \widetilde{\Phi}=A_0^2 +\frac{1}{12}A_2^2 - \frac{\alpha}{4}A_0A_1,\, \widetilde{\Psi}=\frac{1}{216}A_2^3 + \frac{1}{16}A_0A_1^2 -\frac{\alpha}{48}A_0A_1A_2 -\frac{1}{6}A_0^2A_2 +\frac{\alpha^2}{16}A_0^3.
  \end{align*}
  Then, the followings hold:
  \begin{enumerate}[{\rm(i)}]
    \item The divisors $\widetilde{G_2}$ and $\widetilde{G_3}$ intersect transversally at four points on $U_0=\{A_0\neq 0\}$. 
    Moreover, $\widetilde{G_2}$ and $\widetilde{G_3}$ are tangent at $(0:1:0)$. 
    \item The divisors $G_2$ and $G_3$ are not singular on $U_0\subset \CC\PP^2$. 
  \end{enumerate}
\end{proposition}
\begin{proof}
(i)\; 
On $U_0$, the defining equations of $\widetilde{G_2}$ and $\widetilde{G_3}$ are written as
  \begin{align}
    &\widetilde{\Phi} = 1 + \frac{a_2^2}{12}-\frac{\alpha}{4}a_1=0, \label{g2=0}
    \\[4pt]
    &\widetilde{\Psi} = \frac{a_2^3}{216}+\frac{a_1^2}{16}-\frac{\alpha}{48}a_1a_2-\frac{1}{6}a_2+\frac{\alpha^2}{16}=0,\label{g3=0}
  \end{align}
  respectively, where $(a_1,a_2)=(A_1/A_0,A_2/A_0)$. 
  
  If $\alpha\neq0$, we have the equation
  \begin{align*}
    3a_2^4-\alpha^2a_2^3+72a_2^2-108\alpha^2a_2+27\left( \alpha^4+16 \right) = 0.
  \end{align*}
  Set $\displaystyle f(a_2)\coloneq3a_2^4-\alpha^2a_2^3+72a_2^2-108\alpha^2a_2+27\left( \alpha^4+16 \right)$. Then, the discriminant  
  \begin{align*}
    R\left( f,\frac{\dd f}{\dd a_2} \right)=-3^{10}\alpha^4\left( \alpha^2+16 \right)^3\left( \alpha+4 \right)^3\left( \alpha-4 \right)^3
  \end{align*}
  of $f$, where $R(\cdot,\cdot)$ denotes the resultant, does not vanish for almost all $\alpha$. 
  Therefore, $\widetilde{G_2}$ and $\widetilde{G_3}$ intersect transversally at four points on $U_0=\{A_0\neq0\}$ for a generic $\alpha$.
  \footnote{
  This is true even if $\alpha=0$. 
  In fact, if $\alpha=0$, the relations (\ref{g2=0}) and (\ref{g3=0}) can be written as $1+\frac{a_2^2}{12}=0$, $\frac{a_2^3}{216}+\frac{a_1^2}{16}-\frac{1}{6}a_2=0$, which still have four distinct roots.
  }

  On the other hand, on $U_1=\{A_1\neq 0\}$, $\widetilde{G_2}$ and $\widetilde{G_3}$ are defined through the equations 
  \begin{align}
    &u^2+\frac{v^2}{12}-\frac{\alpha}{4}u=0 ,
    \\[4pt]
    &\frac{v^3}{216}+\frac{u}{16}-\frac{\alpha}{48}uv-\frac{1}{6}u^2v+\frac{\alpha^2}{16}u^3=0,
  \end{align}
  respectively, where $(u,v)=(A_0/A_1,A_2/A_1)$. Then, we see that along $u=0$, $\widetilde{G_2}$ and $\widetilde{G_3}$ intersect only at the point $(u,v)=(0,0)$, i.e. $(A_0,A_1,A_2)=(0:1:0)$, where the intersection number is two. Note that neither $\widetilde{G_2}$ nor $\widetilde{G_3}$ contains the point $\left( 0:0:1 \right)\in\CC\PP^2$. Thus, we have proved (i). 

\noindent (ii)\; 
  On $U_0$, the divisor $G_2$ is defined through (\ref{g2=0}). 
  Since $\displaystyle \frac{\del \widetilde{\Phi}}{\del a_1}=-\frac{\alpha}{4}\neq 0$, if $\alpha\neq 0$, $G_2$ is smooth on $U_0$. 
  By (\ref{g3=0}), the singular points of $G_3$ on $U_0$ are characterized by the condition $\displaystyle \left(\widetilde{\Psi}, \frac{\del \widetilde{\Psi}}{\del a_1}, \frac{\del \widetilde{\Psi}}{\del a_2}\right)=\left(0,0,0\right)$, which yields 
  \begin{align}
    \left\{
    \begin{array}{l}
      \displaystyle 8a_2^3-3\alpha^2a_2^2-288a_2+108\alpha^2=0, \\[10pt]
      \displaystyle 4a_2^2-\alpha^2a_2-48=0.
    \end{array}
    \right. \label{eq_3}
    \end{align}
  By an elementary division of polynomials, we have 
  \begin{align*}
    8a_2^3-3\alpha^2a_2^2-288a_2+108\alpha^2=\left( 4a_2^2-\alpha^2a_2-48 \right)\left( 2a_2-\frac{\alpha^2}{4} \right)-\left( \frac{\alpha^4}{4}+192 \right)a_2+96\alpha^2.
  \end{align*}
  This implies that there exists no point satisfying the relations (\ref{eq_3}). Hence, $G_3$ is also smooth on $U_0$.
\end{proof}

\begin{remark}\label{rmk_normal_crossing}
  At a transverse intersection point $p\in\CC\PP^2$ of $G_2$ and $G_3$ in Proposition \ref{prop:intersection}, we can choose $(s_1,s_2)\coloneq \left( g_2^*,g_3^* \right)$ as a local coordinate system of $\CC\PP^2$ centred at $p$, and thus $W$ is locally defined as
\begin{align*}
  Y^2Z=4X^3-s_1XZ^2-s_2Z^3.
\end{align*}
In this case, the discriminant locus is defined through the equation
\begin{align*}
  \Delta=s_1^3-27s_2^2=0.
\end{align*}
As is well known, the discriminant locus has one cusp at $(s_1,s_2)=(0,0)$.
\end{remark}

From this remark and Proposition \ref{prop:Sing}, we obtain the following corollary.

\begin{corollary}\label{singularities of D}
  The singularities of the singular locus $D$ on $\CC\PP^2$ defined through $\Delta=(g_2^*)^3-27(g_3^*)^2=0$ are given as
  \begin{align*}
    {\rm{Sing}}(D) = \left\{\left( 1:\pm\alpha:\frac{\alpha^2}{4}\pm2 \right)\in \CC\PP^2\right\}&\cup\{\left( 0:0:1 \right)\in\CC\PP^2\}\cup\Set{ (0:1:v)\in\CC\PP^2 | v\in\CC }
    \\
    &\cup\Set{p\in \CC\PP^2 | \text{$G_2$ and $G_3$ intersect transversely at $p$}}.
  \end{align*}
\end{corollary}

Further details on the divisor $D$ are described as in the following theorem.

\begin{theorem}\label{Discriminant locus}
  The discriminant locus $D$ consists of a line defined by $A_0=0$ with multiplicity $7$ and a singular quintic curve which has four cusps and two nodes as its singularities. Moreover, this quintic curve is tangent to the line $A_0=0$ at $(0:1:0)$ and $(0:0:1)$.
\end{theorem}
\begin{proof}
  With the homogeneous polynomials $\Phi$ and $\Psi$ as in (\ref{Phi}) and (\ref{Psi}), the discriminant locus $D$ is given as 
  \begin{align}
    0
    =
    \Phi^3-27\Psi^2
    =
    A_0^7\left( \frac{1}{48}A_2^4\,\breve{\Phi}+\frac{1}{4}A_0A_2^2\,\breve{\Phi}^2+A_0^2\,\breve{\Phi}^3-\frac{1}{4}A_2^3\,\breve{\Psi}-27A_0\,\breve{\Psi}^2 \right), \label{eq_dis}
  \end{align}
  where
  \begin{align*}
    \breve{\Phi}=A_0-\frac{\alpha}{4}A_1,~\breve{\Psi}=\frac{1}{16}A_1^2-\frac{\alpha}{48}A_1A_2-\frac{1}{6}A_0A_2+\frac{\alpha^2}{16}A_0^2. 
  \end{align*}
  Since the second factor on the right-hand side of \eqref{eq_dis} is a homogeneous polynomial of degree $5$, the discriminant locus $D$ consists of the line $A_0=0$ with multiplicity $7$ and the quintic curve
  \begin{align}
    \frac{1}{48}A_2^4\,\breve{\Phi}+\frac{1}{4}A_0A_2^2\,\breve{\Phi}^2+A_0^2\,\breve{\Phi}^3-\frac{1}{4}A_2^3\,\breve{\Psi}-27A_0\,\breve{\Psi}^2 =0 \label{quintic_eq}
  \end{align}
  which coincides with the divisor $D$ on $U_0=\{A_0\neq0\}$. From Proposition \ref{prop:intersection} and Corollary \ref{singularities of D}, the four transverse intersection points of $G_2$ and $G_3$ on $U_0$ are cusps of this quintic curve. 
  Moreover, the two points $\displaystyle \left( 1:\pm\alpha:\frac{\alpha^2}{4}\pm2 \right)\in \CC\PP^2$ also belong to the singularities of the quintic curve. 
  We further analyse the quintic curve around the two points $\displaystyle \left( 1:\pm\alpha:\frac{\alpha^2}{4}\pm2 \right)\in \CC\PP^2$. 
  Instead of the affine coordinates $(a_1,a_2)=(A_1/A_0,A_2/A_0)$ on $U_0$, we use the local coordinates $(s_1,s_2)\coloneq(a_1-p,a_2-q)$, where
  \begin{align*}
    (p,q)\coloneq\left( \pm\alpha,\frac{\alpha^2}{4}\pm2 \right).
  \end{align*}
  In a sufficiently small neighbourhood of $(s_1,s_2)=(0,0)$, the discriminant $\Delta$ is written as
  \begin{align*}
    \Delta=\varphi^3+3c_2\varphi^2+3c_2^2\varphi-27\psi^2-54c_3\psi,
  \end{align*}
  where $\varphi$, $\psi$, $c_2$, and $c_3$ are given by
  \begin{align*}
    &\varphi\coloneq \frac{1}{12}\left( s_2^2+2qs_2 \right)-\frac{\alpha}{4}s_1,
    \\[4pt]
    &\psi\coloneq \frac{1}{216}\left( s_2^3+3qs_2^2+3q^2s_2 \right)+\frac{1}{16}\left( s_1^2+2ps_1 \right)-\frac{\alpha}{48}\left( s_1s_2+qs_1+ps_2 \right)-\frac{1}{6}s_2,
    \\[4pt]
    &c_2\coloneq g_2(p,q)=1+\frac{q^2}{12}-\frac{\alpha}{4}p,~c_3\coloneq g_3(p,q)=\frac{q^3}{216}+\frac{p^2}{16}-\frac{\alpha}{48}pq-\frac{1}{6}q+\frac{\alpha^2}{16},\notag
  \end{align*}
  respectively. On the other hand, we can write $\Delta$ as
  \begin{align*}
    \Delta=\Delta_1(s_1,s_2)+\Delta_2(s_1,s_2)+\cdots,
  \end{align*}
  where $\Delta_k(s_1,s_2)$ denotes a homogeneous polynomial of degree $k$ in $s_1,s_2$. 
  Note that the point $(s_1,s_2)=(0,0)$ is the zero of $\varphi$ and $\psi$ both  of order one. 
  We can easily verify that
  \begin{align*}
  \Delta_1(s_1,s_2)=\frac{\del \Delta}{\del s_1}(0,0)\,s_1+\frac{\del \Delta}{\del s_2}(0,0)\,s_2=0.
  \end{align*}
  The term $\Delta_2(s_1,s_2)$ is given as
  \begin{align*}
    \Delta_2(s_1,s_2)=\frac{1}{2!}\left( \frac{\del^2 \Delta}{\del s_1^2}(0,0)\,s_1^2+\frac{\del^2 \Delta}{\del s_2^2}(0,0)\,s_2^2+\frac{\del^2 \Delta}{\del s_1s_2}(0,0)\,s_1s_2 \right).
  \end{align*}
  Considering $p$, $q$, $c_2$, and $c_3$ as polynomials in $\alpha$, we can write $\Delta_2(s_1,s_2)$ as
  \begin{align*}
    \Delta_2(s_1,s_2)=P_1(\alpha)\,s_1^2+P_2(\alpha)\,s_2^2+P_3(\alpha)\,s_1s_2,
  \end{align*}
  where $P_1(\alpha)$, $P_2(\alpha)$, and $P_3(\alpha)$ denote the polynomials in $\alpha$ with $\deg P_1(\alpha)=\deg P_2(\alpha)=4$ and $\deg P_3(\alpha)=5$. 
  Since $P_1(\alpha)P_2(\alpha)$ and $\left( P_3(\alpha)/2 \right)^2$ are the polynomials of degree $8$ and $10$, respectively, we have
  \begin{align*}
    \det
    \left(
      \begin{array}{cc}
        \displaystyle P_1(\alpha) & \displaystyle \frac{P_3(\alpha)}{2} \\
        & \\
        \displaystyle \frac{P_3(\alpha)}{2} & \displaystyle P_2(\alpha) \\
      \end{array}
    \right)
    =P_1(\alpha)P_2(\alpha)-\left( \frac{P_3(\alpha)}{2} \right)^2 \neq  0,
  \end{align*}
  for a generic $\alpha$. 
  This implies that the two singular points $\displaystyle\left( \pm\alpha,\frac{\alpha^2}{4}\pm2 \right)$ are nodes.

  Finally, if $A_0=0$, the left-hand side of (\ref{quintic_eq}) is
  \begin{align*}
    \frac{1}{48}A_2^4\left( -\frac{\alpha}{4}A_1 \right)-\frac{1}{4}A_2^3\left( \frac{1}{16}A_1^2-\frac{\alpha}{48}A_1A_2 \right)
    =-\frac{1}{64}A_1^2A_2^3.
  \end{align*}
  This implies the last statement.
\end{proof}

\begin{notation}
We denote the two components of $D$ which are the line and the singular quintic curve respectively by $L$ and $Q$.
\end{notation}

\subsection{Modification of the base space $\CC\PP^2$ and singular fibres of a smooth model of $W$ as a Miranda elliptic threefold}

As is pointed out in the previous subsection, the Weierstra{\ss} normal form $\pi_{W}\colon W\to\CC\PP^2$ given by (\ref{eq_WM}) does not satisfy the conditions (\ref{condition A}) and (\ref{condition B}). 
In this subsection, we construct a smooth model $\pi_{\widehat{\mathcal{W}}}\colon\widehat{\mathcal{W}}\to \widehat{\CC\PP^2}$ of the elliptic fibration satisfying the following conditions by suitable modifications of the base space $\widehat{\CC\PP^2}\to\CC\PP^2$ and the total space $\widehat{\mathcal{W}}\to W$ along the corresponding locus: 

\begin{itemize}
  \item The Weierstra{\ss} normal form $\pi_{\widehat{\mathcal{W}}}\colon\widehat{\mathcal{W}}\to \widehat{\CC\PP^2}$ is birationally equivalent to $\pi_{W}\colon W\to\CC\PP^2$.
  \item The elliptic fibration $\pi_{\widehat{\mathcal{W}}}\colon\widehat{\mathcal{W}}\to\widehat{\CC\PP^2}$ satisfies the condition (\ref{condition A})--(C) in Subsection \ref{Miranda's elliptic threefolds}.
  \item All the colliding types appearing in $\widehat{\mathcal{W}}$ are on Miranda's list.
\end{itemize} 
Then, we study the singular fibres of the elliptic fibration $\pi_{\widehat{\mathcal{W}}}\colon \widehat{\mathcal{W}}\to\widehat{\CC\PP^2}$.

\subsubsection{Modification of the base plane $\CC\PP^2$}\label{modification of the base space}

Here, we change the base space $\CC\PP^2$ through birational mappings to make the reduced total transforms of the reduced divisors $A_0$, $B_0$, $D_0$ satisfy the conditions (\ref{condition A})--(C) in Subsection \ref{Miranda's elliptic threefolds}. 
In what follows, we construct the birational changes according to the three different types of singular points of $D_0$.

\begin{enumerate}[(a)]
  \item {Blowing-ups at the four transverse intersections of $G_2$ and $G_3$:}\\
  We consider a transverse intersection point $p\in \mathrm{supp}(D)$ of $G_2$ and $G_3$ on $U_0=\{A_0\neq0\}$. 
  As in Remark \ref{rmk_normal_crossing}, regarding $(s_1,s_2)\coloneq (g_2^*,g_3^*)$ as a local coordinate system of $\CC\PP^2$ with the centre at $p$, the discriminant locus $D$ is given as $\displaystyle\Delta=s_1^3-27s_2^2=0$. 
  On $U_0$, the discriminant locus $D$ coincides with the singular quintic curve $Q$ and the point $p$ is a cusp point. 
  As is well known, by blowing up three times, all the singular points of the reduced total transform of $Q$ turn to be nodes. 
  Moreover, the local equations of the total and the proper transforms of $Q$ denoted by $\widehat{\Delta}$ and $\widetilde{\Delta}$ respectively are written in different open sets as
  \begin{align}
    \begin{cases}
      \widehat{\Delta}={s_1'}^2{s_2'}^6\left( s_1'-27 \right),
      \\[4pt]
      \widehat{\Delta}=\bar{s}_1^6\,\bar{s}_2^3\left( 1-27\bar{s}_2 \right),
    \end{cases}
    \begin{cases}
      \widetilde{\Delta}=s_1'-27,
      \\[4pt]
      \widetilde{\Delta}= 1-27\bar{s}_2,
    \end{cases}
    \label{cusp:total transform of D }
  \end{align}
  where $(s_1',s_2')$ and $(\bar{s}_1,\bar{s}_2)$ are the blowing-up coordinates satisfying the relations $s_1's_2'=\bar{s}_1$ and $\bar{s}_1\bar{s}_2=s_2'$. 
  We denote the smooth curves $s_1'=0$, $\bar{s}_2=0$, and $\widetilde{\Delta}=0$ respectively by $E_1$, $E_2$, and $\widetilde{Q}$. 
  The exceptional curve $E_3$ with respect to the third blowing-up is defined by $s_2'=0$ or $\bar{s}_1=0$. 
  Note that, if we denote the whole birational change up to here by $\sigma_{p}$, we have $\displaystyle\sigma_{p}^{-1}(p)=E_1\cup E_2\cup E_3$.
  The reduced total transform of $Q$ has three nodes as follows:
  \begin{itemize}
    \item The intersection point $s_1'=s_2'=0$ of $E_1$ and $E_3$, denoted by $r_1$ in Figure \ref{blowing up the cusp}.
    \item The intersection point $\bar{s}_1=\bar{s}_2=0$ of $E_2$ and $E_3$, denoted by $r_2$ in Figure \ref{blowing up the cusp}.
    \item The intersection point $s_1'=27,\,s_2'=0$ or equivalently $\bar{s}_1=0,\,\bar{s}_2=1/27$ of $\widetilde{Q}$ and $E_3$, denoted by $r_3$ in Figure \ref{blowing up the cusp}.
  \end{itemize}
  
  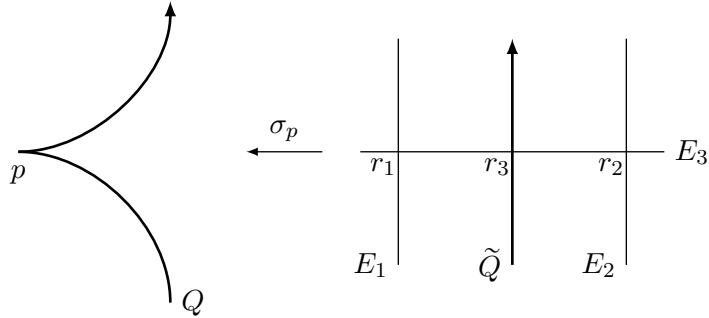
\begin{figure}[H]
    \centering
    \begin{tikzpicture}[>=latex]
      \fill[black] 
      (0,-0.3) circle(0) node {$p$} 
      (4.8,-0.2) circle(0) node {$r_1$} 
      (6.3,-0.2) circle(0) node {$r_3$} 
      (7.8,-0.2) circle(0) node {$r_2$};
      \draw[line width=1pt,->] (0,0) .. controls (1,0) and (2,1) .. (2,2); 
      \draw[line width=1pt] (0,0) ..controls (1,0) and (2,-1) .. (2,-2) node[right] {$Q$};
      \draw[line width=0.5pt,<-] (3,0) -- (4,0) node[midway,above] {$\sigma_{p}$};
      \draw[line width=0.5pt] (4.5,0) -- (8.5,0) node[right] {$E_3$};
      \draw[line width=0.5pt] (5,1.5) -- (5,-1.5) node[left] {$E_1$};
      \draw[line width=0.5pt] (8,1.5) -- (8,-1.5) node[left] {$E_2$};
      \draw[line width=1pt,<-] (6.5,1.5) -- (6.5,-1.5) node[left] {$\widetilde{Q}$};
  \end{tikzpicture}
  \caption{Blowing up at the cusp}
  \label{blowing up the cusp}
  \end{figure}

  \item {\bf Blowing-ups at $(0:1:0)$:}\\
  By using the inhomogeneous coordinates $(u,v)=(A_0/A_1,A_2/A_1)$ on $U_1=\{A_1\neq0\}$, the discriminant locus $D$ is given by
  \begin{align*}
    \Delta=u^7\left( \frac{v^4}{48}\varphi^{\prime}+\frac{v^2}{4}u\left(\varphi^{\prime}\right)^2+u^2\left(\varphi^{\prime}\right)^3-\frac{v^3}{4}\psi^{\prime}-27u\left(\psi^{\prime}\right)^2 \right)=0,
  \end{align*}
  where 
  \begin{align*}
    \varphi^{\prime}= u-\frac{\alpha}{4},~\displaystyle\psi^{\prime}=\frac{1}{16}-\frac{\alpha}{48}v-\frac{1}{6}uv+\frac{\alpha^2}{16}u^2.
  \end{align*}

  Since the quintic curve $Q$ defined by $\displaystyle \frac{v^4}{48}\varphi^{\prime}+\frac{v^2}{4}u\left(\varphi^{\prime}\right)^2+u^2\left(\varphi^{\prime}\right)^3-\frac{v^3}{4}\psi^{\prime}-27u\left(\psi^{\prime}\right)^2=0$ on $U_1$ is tangent to the line $L$ defined by $u=0$ at the origin and $Q$ and $L$ do not intersect on $U_1$ except for the origin, we only have to blow up at the origin. 
  Through the two successive blowing-ups at the origin, the local equation of the total transform of $D$, which we denote by $\Delta'$, is written as
  \begin{align*}
    \left\{
      \begin{array}{l}
      \displaystyle\Delta'= {u'}^{16}{v'}^8\left( \frac{{u'}^2{v'}^3}{48}\varphi_1'+\frac{\left( u'v' \right)^2}{4}\left( \varphi_1' \right)^2+{u'}^2v'\left( \varphi_1' \right)^3-\frac{u'{v'}^2}{4}\psi_1'-27\left( \psi_1' \right)^2 \right),
      \\[15pt]
      \displaystyle\Delta'= \bar{u}^7\bar{v}^{16}\left( \frac{\bar{v}^2}{48}\bar{\varphi}_1+\frac{\bar{u}\bar{v}^2}{4}\left( \bar{\varphi}_1 \right)^2+\left( \bar{u}\bar{v} \right)^2\left( \bar{\varphi}_1 \right)^3-\frac{\bar{v}}{4}\bar{\psi}_1-27\bar{u}\left( \bar{\psi}_1 \right)^2 \right),
      \end{array}
    \right.
  \end{align*}
  where $(u',v')$ and $(\bar{u},\bar{v})$ are the blowing-up coordinates and $\varphi_1'$, $\psi_1'$, $\bar{\varphi}_1$, and $\bar{\psi}_1$ are given as 
  \begin{align*}
    \varphi_1'={u'}^2v'-\frac{\alpha}{4},~\psi_1'=\frac{1}{16}-\frac{\alpha}{48}u'v'-\frac{1}{6}{u'}^3{v'}^2+\frac{\alpha^2}{16}{u'}^4{v'}^2,
  \end{align*}
  or 
  \begin{align*}
    \bar{\varphi}_1=\bar{u}\bar{v}^2-\frac{\alpha}{4},~\bar{\psi}_1=\frac{1}{16}-\frac{\alpha}{48}\bar{v}-\frac{1}{6}\bar{u}\bar{v}^3+\frac{\alpha^2}{16}\bar{u}^2\bar{v}^4.
  \end{align*}
  We denote the divisors $v'=0$, $u'=0$ (or $\bar{v}=0$), and $\bar{u}=0$ respectively by $E_1$, $E_2$, and $L'$.
  Note that $E_2$ is the exceptional divisor with respect to the second blowing-up. 
  We further consider the divisor $Q'$ defined through
  \begin{align}
    \frac{{u'}^2{v'}^3}{48}\varphi_1'+\frac{\left( u'v' \right)^2}{4}\left( \varphi_1' \right)^2+{u'}^2v'\left( \varphi_1' \right)^3-\frac{u'{v'}^2}{4}\psi_1'-27\left( \psi_1' \right)^2=0, \label{eq_5}
  \end{align}
  or
  \begin{align}
    \frac{\bar{v}^2}{48}\bar{\varphi}_1+\frac{\bar{u}\bar{v}^2}{4}\left( \bar{\varphi}_1 \right)^2+\left( \bar{u}\bar{v} \right)^2\left( \bar{\varphi}_1 \right)^3-\frac{\bar{v}}{4}\bar{\psi}_1-27\bar{u}\left( \bar{\psi}_1 \right)^2=0. \label{eq_6}
  \end{align}
  By plugging in $u'=0$ for the left-hand side of the equation (\ref{eq_5}), we have
  \begin{align*}
    -27\psi_1'^2=-26\left( \frac{1}{16} \right)^2\neq0,
  \end{align*}
  that is, $\displaystyle\bar{\Delta'}(u',v')\coloneq\frac{{u'}^2{v'}^3}{48}\varphi_1'+\frac{\left( u'v' \right)^2}{4}\left( \varphi_1' \right)^2+{u'}^2v'\left( \varphi_1' \right)^3-\frac{u'{v'}^2}{4}\psi_1'-27\left( \psi_1' \right)^2$ is a unit at $(u',v')=(0,0)$, where $E_1$ and $E_2$ intersect.
  On the other hand, if $\bar{v}=0$, the left-hand side of the equation (\ref{eq_6}) is 
  \begin{align*}
    -27\bar{u}\bar{\psi}_1^2=-26\left( \frac{1}{16} \right)\bar{u}.
  \end{align*}
  Since the multiplicity of (\ref{eq_6}) is equal to one, the singularities of the reduced total transform of $D$ consist of
  \begin{itemize}
    \item The triple point $\bar{u}=\bar{v}=0$, denoted by $r'_1$ in Figure \ref{blowing up at infinity}.
    \item The intersection point $u'=v'=0$ of $E_1$ and $E_2$, denoted by $r'_2$ in Figure \ref{blowing up at infinity}.
  \end{itemize}
  \begin{figure}[H]
    \centering
    \begin{tikzpicture}[>=latex]
      \fill[black] 
      (-1.7,-0.3) circle(0) node {$r'_1$} 
      (1.8,-0.3) circle(0) node {$r'_2$};
      \draw[line width=0.5pt] (-3,0) -- (3,0) node[right] {$E_2$};
      \draw[line width=0.5pt] (2,1.5) -- (2,-1.5) node[right] {$E_1$};
      \draw[line width=1pt,<-] (-1.5,1.5) -- (-1.5,-1.5) node[right] {$L'$};
      \draw[line width=1pt,<-] (-2.5,1.3) .. controls (-2.2,1) .. (-1.5,0);
      \draw[line width=1pt] (-1.5,0) .. controls (-0.8,-1) .. (-0.5,-1.3) node[right] {$Q'$};
    \end{tikzpicture}
    \caption{Two blowing-ups at $(0:1:0)$}
    \label{blowing up two times at infinity}
    \label{blowing up at infinity}
  \end{figure}
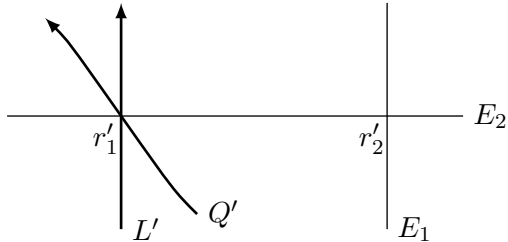
  Furthermore, we blow up the base space at $r'_1$ so that the singularities of the reduced total transform of $D$ are only nodes and at $r'_2$ so that all the colliding types are on Miranda's list. Then, the reduced total transform of $D$ has only nodes as singularities. The local equation $\widehat{\Delta}$ of the total transform of $D$ is written as
  \begin{align*}
    {\text{around}\; r_1^{\prime}:}
    &\left\{
    \begin{array}{l}
      \displaystyle\widehat{\Delta}= \bar{u}^{24}s^{16}\left( \frac{s^2}{48}\bar{u}\varphi_2'+\frac{\left( \bar{u}s \right)^2}{4}{\left( \varphi_2' \right)}^2+\bar{u}^3s^2{\left( \varphi_2' \right)}^3-\frac{s}{4}\psi_2'-27{\left( \psi_2' \right)}^2 \right),
      \\[15pt]
      \displaystyle\widehat{\Delta}= t^7\bar{v}^{24}\left( \frac{\bar{v}}{48}\bar{\varphi}_2+\frac{\bar{v}^2}{4}t{\left( \bar{\varphi}_2 \right)}^2+t^2\bar{v}^3{\left( \bar{\varphi}_2 \right)}^3-\frac{1}{4}\bar{\psi}_2-27t{\left( \bar{\psi}_2 \right)}^2 \right),
      \end{array}
    \right.
    \\[3pt]
    {\text{around}\; r_2^{\prime}:}
    &\left\{
      \begin{array}{l}
      \displaystyle\widehat{\Delta}= {u'}^{24}{s'}^8\,\bar{\Delta'}(u',s'u'),
      \\[4pt]
      \displaystyle\widehat{\Delta}= {t'}^{16}{v'}^{24}\,\bar{\Delta'}(t'v',v'),
      \end{array}
    \right.
  \end{align*}
  where $(u',s')$, $(t',v')$, $(\bar{u},s)$, and $(t,\bar{v})$ are the blowing-up coordinates and $\varphi_2'$, $\psi_2'$, $\bar{\varphi}_2$, and $\bar{\psi}_2$ are given as
  \begin{align*}
    \varphi_2'=\bar{u}^3s^2-\frac{\alpha}{4},~\psi_2'=\frac{1}{16}-\frac{\alpha}{48}\bar{u}s-\frac{1}{6}\bar{u}^4s^3+\frac{\alpha^2}{16}\bar{u}^6s^4,
  \end{align*}
  or
  \begin{align*}
    \bar{\varphi}_2=t\bar{v}^3-\frac{\alpha}{4},~\bar{\psi}_2=\frac{1}{16}-\frac{\alpha}{48}\bar{v}-\frac{1}{6}t\bar{v}^4+\frac{\alpha^2}{16}t^2\bar{v}^6.
  \end{align*}
  The exceptional divisors $E_3$ and $E_4$ are defined through the equations $\bar{u}=0$ (or $\bar{v}=0$) and $u'=0$ (or $v'=0$), respectively. The proper transforms of $E_1$, $E_2$, and $L'$ are defined through the equations $s'=0$, $s=0$ (or $t'=0$), and $t=0$, respectively. Moreover, the proper transform of $Q'$ is given by the equation
  \begin{align*}
    \frac{s^2}{48}\bar{u}\varphi_2'+\frac{\left( \bar{u}s \right)^2}{4}{\left( \varphi_2' \right)}^2+\bar{u}^3s^2{\left( \varphi_2' \right)}^3-\frac{s}{4}\psi_2'-27{\left( \psi_2' \right)}^2=0,
  \end{align*}
  or
  \begin{align*}
    \frac{\bar{v}}{48}\bar{\varphi}_2+\frac{\bar{v}^2}{4}t\left( \bar{\varphi}_2 \right)^2+t^2\bar{v}^3\left( \bar{\varphi}_2 \right)^3-\frac{1}{4}\bar{\psi}_2-27t\left( \bar{\psi}_2 \right)^2=0.
  \end{align*}
  Now, if we denote the proper transforms of $E_1$, $E_2$, $L'$, and $Q'$ by $\widetilde{E}_1$, $\widetilde{E}_2$, $\widetilde{L}$, and $\widetilde{Q}$, respectively, the configuration of these divisors at infinity is as in Figure \ref{ultimate blowing up at (0:1:0)}. 
  \begin{figure}[H]
    \centering
    \begin{tikzpicture}[>=latex]
      \draw[line width=1pt,<-] (-12,-1) -- (-8,-1) node[right] {$L$};
      \draw[line width=1pt,<-] (-12,-2) to [out=45,in=180] (-10,-1);
      \draw[line width=1pt] (-8,0) node[right] {$Q$} to [out=225,in=360] (-10,-1);
      \draw[line width=0.5,<-] (-6.5,-1) -- (-4.5,-1) node[midway,above] {blowing-ups};
      \draw[line width=0.5pt] (-3,0) node[left] {$\widetilde{E_2}$} -- (3,0) ;
      \draw[dashed] (2,1.75)  -- (2,-0.5) node[right] {$E_4$};
      \draw[line width=0.5pt] (1,1.5) -- (3,1.5) node[right] {$\widetilde{E}_1$};
      \draw[dashed] (-1.5,1) node[right] {$E_3$} -- (-1.5,-3.5);
      \draw[line width=1pt,<-] (-3,-1.5) .. controls (-2.3,-1.4) .. (-1.5,-1.5);
      \draw[line width=1pt] (0,-1.5) node[right] {$\widetilde{Q}$} .. controls (-0.7,-1.6) .. (-1.5,-1.5);
      \draw[line width=1pt,<-] (-3,-3) -- (0,-3) node[right] {$\widetilde{L}$};
    \end{tikzpicture}
    \caption{Divisors after the blowing-ups at $(0:1:0)$}
    \label{ultimate blowing up at (0:1:0)}
  \end{figure}
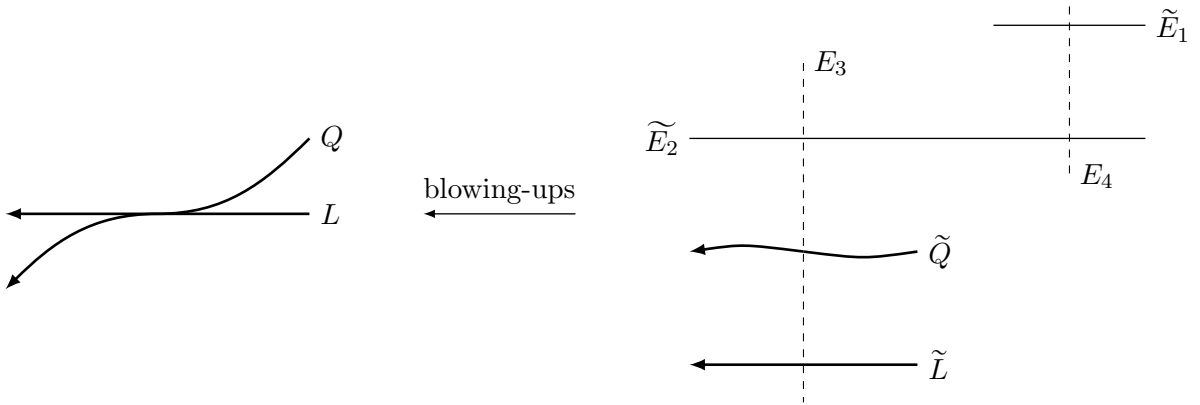
  
  \item {\bf Blowing-ups at $(0:0:1)$:}\\
  By using the inhomogeneous coordinates $(u,v)\coloneq(A_0/A_2,A_1/A_2)$ on $U_2=\{A_2\neq0\}$, the local equation of the discriminant locus $D$ is given by
  \begin{align*}
    \Delta=u^7\left( \frac{1}{48}\varphi+\frac{1}{4}u\varphi^2+u^2\varphi^3-\frac{1}{4}\psi-27u\psi^2 \right),
  \end{align*}
  where
  \begin{align*}
    \varphi=u-\frac{\alpha}{4}v,~\psi=\frac{1}{16}v^2-\frac{\alpha}{48}v-\frac{1}{6}u+\frac{\alpha^2}{16}u^2.
  \end{align*}
  The quintic curve $Q$ defined by $\displaystyle\frac{1}{48}\varphi+\frac{1}{4}u\varphi^2+u^2\varphi^3-\frac{1}{4}\psi-27u\psi^2=0$ is tangent to the line $L$ defined by $u=0$ at $(u,v)=(0,0)$, i.e. at $(0:0:1)\in\CC\PP^2$, with the intersection number $2$ at this point. Therefore, we need to blow up at $(u,v)=(0,0)$ and repeat this procedure so that the singularities of the reduced discriminant locus consist only of nodes. We can easily verify that this condition is satisfied after the two blowing-ups. Then, the local equation of the total transform of $D$ is written as
  \begin{align*}
    \left\{
      \begin{array}{l}
      \displaystyle\widehat{\Delta}= {u'}^{16}{v'}^8\left( \frac{1}{48}+\frac{1}{4}{u'}^2{v'}^2{\left( \varphi' \right)}^2+{u'}^6{v'}^4{\left( \varphi' \right)}^3-\frac{1}{4}\left( \frac{1}{16}v'-\frac{1}{6}+\frac{\alpha^2}{16}{u'}^2v' \right)-27{u'}^2{v'}^2\psi' \right),
      \\[15pt]
    \displaystyle\widehat{\Delta}= \bar{u}^7\bar{v}^{16}\left( \frac{1}{48}\bar{u}+\frac{1}{4}\bar{u}\bar{v}^2\left( \bar{\varphi} \right)^2+\bar{u}^2\bar{v}^5\left( \bar{\varphi} \right)^3-\frac{1}{4}\left( \frac{1}{16}-\frac{1}{6}\bar{u}+\frac{\alpha^2}{16}\bar{u}^2\bar{v}^2 \right) -27\bar{u}\bar{v}^2\bar{\psi} \right),
      \end{array}
    \right.
  \end{align*}
  where $(u',v')$ and $(\bar{u},\bar{v})$ are the blowing-up coordinates and $\varphi'$, $\psi'$, $\bar{\varphi}$, and $\bar{\psi}$ are given as
  \begin{align*}
    \varphi'=u'-\frac{\alpha}{4},~\psi'=\frac{1}{16}u'v'-\frac{\alpha}{48}-\frac{1}{6}u'+\frac{\alpha^2}{16}{u'}^3v',
  \end{align*}
  or
  \begin{align*}
    \bar{\varphi}=\bar{u}\bar{v}-\frac{\alpha}{4},~\bar{\psi}=\frac{1}{16}\bar{v}-\frac{\alpha}{48}-\frac{1}{6}\bar{u}\bar{v}+\frac{\alpha^2}{16}\bar{u}^2\bar{v}^3.
  \end{align*}
  By $E_1$, $E_2$, and $\widetilde{L}$, we denote the divisors defined through $v'=0$, $u'=0$ (or $\bar{v}=0$), and $\bar{u}=0$, respectively. 
  Moreover, we consider the divisor $\widetilde{Q}$ with the local equation given by 
  \begin{align*}
      \displaystyle\widetilde{\Delta}=  \frac{1}{48}+\frac{1}{4}{u'}^2{v'}^2{\left( \varphi' \right)}^2+{u'}^6{v'}^4{\left( \varphi' \right)}^3-\frac{1}{4}\left( \frac{1}{16}v'-\frac{1}{6}+\frac{\alpha^2}{16}{u'}^2v' \right)-27{u'}^2{v'}^2\psi'=0,
  \end{align*}
  or
  \begin{align*}
    \displaystyle\widetilde{\Delta}=  \frac{1}{48}\bar{u}+\frac{1}{4}\bar{u}\bar{v}^2\left( \bar{\varphi} \right)^2+\bar{u}^2\bar{v}^5\left( \bar{\varphi} \right)^3-\frac{1}{4}\left( \frac{1}{16}-\frac{1}{6}\bar{u}+\frac{\alpha^2}{16}\bar{u}^2\bar{v}^2 \right) -27\bar{u}\bar{v}^2\bar{\psi}=0.
  \end{align*}
  The singularities of the reduced total transform of $D$ can be described as follows:
  \begin{itemize}
    \item The intersection point $u'=v'=0$ of $E_1$ and $E_2$, denoted by $q_1$ in Figure \ref{blowing up at (0:0:1)}.
    \item The intersection point $(u',v')=(0,4)$ or equivalently $(\bar{u},\bar{v})=(1/4,0)$ of $E_2$ and $\widetilde{Q}$, denoted by $q_2$ in Figure \ref{blowing up at (0:0:1)}.
    \item The intersection point $\bar{u}=\bar{v}=0$ of $E_2$ and $\widetilde{L}$, denoted by $q_3$ in Figure \ref{blowing up at (0:0:1)}.
  \end{itemize}

  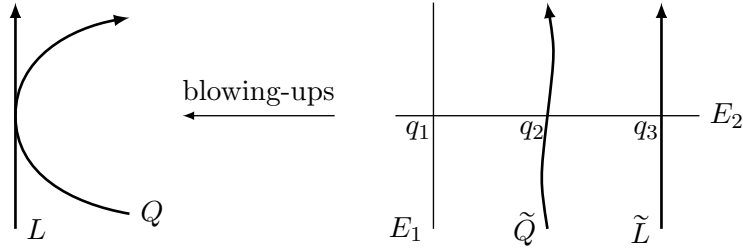
\begin{figure}[H]
    \centering
    \begin{tikzpicture}[>=latex]
      \fill[black] 
      (4.8,-0.2) circle(0) node {$q_1$} 
      (6.3,-0.2) circle(0) node {$q_2$} 
      (7.8,-0.2) circle(0) node {$q_3$};
      \draw[line width=1pt,<-] (-.5,1.5) -- (-.5,-1.5) node[right] {$L$};
      \draw[line width=1pt,<-] (1,1.3) to [out=190,in=90] (-.5,0);
      \draw[line width=1pt] (1,-1.3) node[right] {$Q$} to [out=170,in=270] (-.5,0);
      \draw[line width=0.5pt,<-] (1.7,0) -- (3.7,0) node[midway,above] {blowing-ups}; 
      \draw[line width=0.5pt] (4.5,0) -- (8.5,0) node[right] {$E_2$};
      \draw[line width=0.5pt] (5,1.5) -- (5,-1.5) node[left] {$E_1$};
      \draw[line width=1pt,<-] (8,1.5) -- (8,-1.5) node[left] {$\widetilde{L}$};
      \draw[line width=1pt,<-] (6.5,1.5) .. controls (6.6,0.85) .. (6.5,0);
      \draw[line width=1pt] (6.5,0) .. controls (6.4,-0.85) .. (6.5,-1.5) node[left] {$\widetilde{Q}$};
  \end{tikzpicture}
  \caption{Blowing up at $(0:0:1)$}
  \label{blowing up at (0:0:1)}
  \end{figure}
  
\end{enumerate}

Through the above procedures, we obtain the desired base change $\widehat{\CC\PP^2}\to\CC\PP^2$.

\subsubsection{Construction of the smooth model $\widehat{\mathcal{W}}$ and its singular fibres}

Via the above base change $\widehat{\CC\PP^2}\to\CC\PP^2$, taking the pull-back of the original elliptic fibration $\pi_{W}\colon W\to\CC\PP^2$, we obtain an elliptic fibration in Weierstra{\ss} normal form $\pi_{\mathcal{W}}\colon\mathcal{W}\to\widehat{\CC\PP^2}$, which is birationally equivalent to the original elliptic fibration $\pi_{W}\colon W\to\CC\PP^2$. At this stage, the smooth model $\pi_{\widehat{\mathcal{W}}}\colon\widehat{\mathcal{W}}\to \widehat{\CC\PP^2}$ can easily be obtained from $\mathcal{W}$ through blowing-ups of the total space. In each of three cases discussed in Sub-subsection 4.3.1, the singular fibres of $\pi_{\widehat{\mathcal{W}}}$ are determined as follows:

\begin{enumerate}[(a)]
  \item The defining equation of $W$ on a neighbourhood of each of the four transverse intersection points of $G_2$ and $G_3$ is written in the form $\displaystyle Y^2Z=4X^3-s_1XZ^2-s_2Z^3$, and thus $G_2$ and $G_3$ are defined through the equations $s_1=0$ and $s_2=0$, respectively. Then, the total transforms of $G_2$ and $G_3$ are defined by 
  \begin{align}
    \widehat{G_2}:
    \begin{cases}
      s_1'{s_2'}^2=0,
      \\[4pt]
      \bar{s}_1^2\bar{s}_2=0,
    \end{cases}
  \end{align}
  and
  \begin{align}
    \widehat{G_3}:
    \begin{cases}
      s_1'{s_2'}^3=0,
      \\[4pt]
      \bar{s}_1^3\bar{s}_2^2=0,
    \end{cases}
  \end{align}
  respectively, where $(s_1',s_2')$ and $(\bar{s}_1,\bar{s}_2)$ are the blowing-up coordinates. Furthermore, the following equations are the pull-back of the defining equation of $W$:
  \begin{align}
    \begin{cases}
      Y^2Z=4X^3-s_1'{s_2'}^2XZ^2-s_1'{s_2'}^3Z^3,
      \\[4pt]
      Y^2Z=4X^3-\bar{s}_1^2\bar{s}_2XZ^2-\bar{s}_1^3\bar{s}_2^2Z^3.
    \end{cases}
  \end{align}
  These equations are also the local equations of $\mathcal{W}$. By blowing-up $\mathcal{W}$ along $E_1\cup E_2\cup E_3$ appearing in Figure \ref{blowing up the cusp}, we obtain a local desingularization of $\mathcal{W}$. 
  Note that the discriminant is given by (\ref{cusp:total transform of D }). 
  The singular fibres are described as follows:
  \begin{itemize}
    \item The singular fibres over generic points of $E_1$ are of type $II$ since $(L,K,N)=(1,1,2)$ over $E_1$.
    \item The singular fibres over generic points of $E_2$ are of type $III$ since $(L,K,N)=(1,2,3)$ over $E_2$.
    \item The singular fibres over generic points of $E_3$ are of type $I_0^*$ since $(L,K,N)=(2,3,6)$ over $E_3$.
    \item The singular fibres over generic points of $\widetilde{Q}$ are of type $I_1$ since $(L,K,N)=(0,0,1)$ over $\widetilde{Q}$.
    \item The singular fibres over the collision points $r_1$, $r_2$, and $r_3$ in Figure \ref{blowing up the cusp}, which are the intersection point of $E_1$ with $E_3$, the one $E_2$ with $E_3$, and the one $\widetilde{Q}$ with $E_3$, respectively, are not of the Kodaira types. 
    On Miranda's list, the dual graphs of each singular fibre are written as in Figures \ref{The singular fibre over p_1}, \ref{The singular fibre over p_2}, \ref{The singular fibre over p_3}.
    \begin{figure}[H]
      \centering
      \begin{minipage}{0.4\textwidth}
        \centering
        \scalebox{0.8}[0.8]{
        \begin{tikzpicture}[line width=1pt, node distance=1.3cm]
          \node(a)           [draw,circle]{$1$};
          \node[right of=a](b)   [draw,circle]{$2$};
          \node[right of=b](c)   [draw,circle]{$3$};
          \draw (a) -- (b);
          \draw (b) -- (c);
        \end{tikzpicture}
        }
        \caption{The singular fibre over $r_1$}
        \label{The singular fibre over p_1}
      \end{minipage}
      \begin{minipage}{0.4\textwidth}
        \centering
        \scalebox{0.8}[0.8]{
        \begin{tikzpicture}[line width=1pt, node distance=1.3cm]
          \node(a)           [draw,circle]{$1$};
          \node[right of=a](b)   [draw,circle]{$2$};
          \node[right of=b](c)   [draw,circle]{$3$};
          \node[right of=c](d)           [draw,circle]{$2$};
          \node[right of=d](f)   [draw,circle]{$1$};
          \draw (a) -- (b);
          \draw (b) -- (c);
          \draw (c) -- (d);
          \draw (d) -- (f);
      \end{tikzpicture}
        }
        \caption{The singular fibre over $r_2$}
        \label{The singular fibre over p_2}
      \end{minipage}
      \begin{minipage}{0.4\textwidth}
        \centering
        \vspace{2em}
        \scalebox{0.8}[0.8]{
        \begin{tikzpicture}[line width=1pt, node distance=1.3cm]
          \node(a)           [draw,circle]{$1$};
          \node[right of=a](b)   [draw,circle]{$2$};
          \node[right of=b](c)   [draw,circle]{$1$};
          \draw (a) -- (b);
          \draw (b) -- (c);
        \end{tikzpicture}
        }
        \caption{The singular fibre over $r_3$}
        \label{The singular fibre over p_3}
      \end{minipage}
    \end{figure}
  \end{itemize}
  \item After blowing-up at $(0:1:0)\in\CC\PP^2$, the local equations of the total transforms of $G_2$ and $G_3$ are defined by
  \begin{align}
    \widehat{G_2}:
    \left\{
    \begin{array}{l}
      \displaystyle {u'}^9{s'}^3\left( \frac{u's'}{12}+\varphi_1'(u',s'u') \right)=0,
      \\[10pt]
      \displaystyle {t'}^6{v'}^9\left( \frac{v'}{12}+\varphi_1'(t'v',v') \right)=0,
      \\[10pt]
      \displaystyle \bar{u}^8s^6\left( \frac{1}{12}+\bar{u}\varphi_2' \right)=0,
      \\[10pt]
      \displaystyle t^2\bar{v}^8\left( \frac{1}{12}+t\bar{v}\bar{\varphi}_2 \right)=0,
    \end{array}
    \right.
  \end{align}
  and
  \begin{align}
    \widehat{G_3}:
    \left\{
    \begin{array}{l}
      \displaystyle {u'}^{12}{s'}^4\left( \frac{{u'}^3{s'}^2}{216}+\psi_1'(u',s'u') \right)=0,
      \\[10pt]
      \displaystyle {t'}^8{v'}^{12}\left( \frac{t'{v'}^3}{216}+\psi_1'(t'v',v') \right)=0,
      \\[10pt]
      \displaystyle \bar{u}^{12}s^8\left( \frac{s}{216}+\psi_2' \right)=0,
      \\[10pt]
      \displaystyle t^3\bar{v}^{12}\left( \frac{1}{216}+t\bar{\psi}_2 \right)=0,
    \end{array}
    \right.
  \end{align}

  respectively, where $(u',s')$, $(t',v')$, $(\bar{u},s)$, and $(t,\bar{v})$ are the blowing-up coordinates. Furthermore, the following equations are the pull-back of the defining equation of $W$:
  \begin{align*}
    \left\{
      \begin{array}{l}
        \displaystyle Y^2Z=4X^3-{u'}^9{s'}^3\left( \frac{u's'}{12}+\varphi_1'(u',s'u') \right)XZ^2-{u'}^{12}{s'}^4\left( \frac{{u'}^3{s'}^2}{216}+\psi_1'(u',s'u') \right)Z^3,
        \\[10pt]
        \displaystyle Y^2Z=4X^3-{t'}^6{v'}^9\left( \frac{v'}{12}+\varphi_1'(t'v',v') \right)XZ^2-{t'}^8{v'}^{12}\left( \frac{t'{v'}^3}{216}+\psi_1'(t'v',v') \right)Z^3,
        \\[10pt]
        \displaystyle Y^2Z=4X^3-\bar{u}^8s^6\left( \frac{1}{12}+\bar{u}\varphi_2' \right)XZ^2-\bar{u}^{12}s^8\left( \frac{s}{216}+\psi_2' \right)Z^3,
        \\[10pt]
        \displaystyle Y^2Z=4X^3-t^2\bar{v}^8\left( \frac{1}{12}+t\bar{v}\bar{\varphi}_2 \right)XZ^2-t^3\bar{v}^{12}\left( \frac{1}{216}+t\bar{\psi}_2 \right)Z^3.
      \end{array}
    \right.
  \end{align*}
  
  Since these local equations do not satisfy the condition (C) in Subsection \ref{Miranda's elliptic threefolds}, we change the homogeneous fibre coordinates as described in the subsection \ref{Miranda's elliptic threefolds}. Then, we obtain the local equations of $\mathcal{W}$ and the discriminant as follows:
  \begin{align}
    \left\{
        \begin{array}{l}
          \displaystyle Y^2Z=4X^3-u'{s'}^3\,\bar{\Phi}_1XZ^2-{s'}^4\,\bar{\Psi}_1Z^3,
          \\[4pt]
          \displaystyle Y^2Z=4X^3-{t'}^2v'\,\bar{\Phi}_1XZ^2-{t'}^2\,\bar{\Psi}_1Z^3,
          \\[4pt]
          \displaystyle Y^2Z=4X^3-s^2\,\bar{\Phi}_2XZ^2-s^2\,\bar{\Psi}_2Z^3,
          \\[4pt]
          \displaystyle Y^2Z=4X^3-t^2\,\bar{\Phi}_2XZ^2-t^3\,\bar{\Psi}_2Z^3,
        \end{array}
      \right.
      \hspace{2em}
      \left\{
        \begin{array}{l}
        \displaystyle\widehat{\Delta}= {s'}^8\,\bar{\Delta},
        \\[4pt]
        \displaystyle\widehat{\Delta}= {t'}^{4}\,\bar{\Delta},
        \\[4pt]
        \displaystyle\widehat{\Delta}= s^{4}\widetilde{\Delta},
        \\[4pt]
        \displaystyle\widehat{\Delta}= t^7\widetilde{\Delta}.
        \end{array}
      \right.
  \end{align}
  By blowing-up $\mathcal{W}$ along $\widetilde{E}_1\cup\widetilde{E}_2\cup\widetilde{L}$ appearing in Figure \ref{ultimate blowing up at (0:1:0)}, we have a local desingularization of $\mathcal{W}$. Note that the fibres over $E_3\cup E_4$ except for the intersection points with $\widetilde{Q}$, $\widetilde{L}$, $\widetilde{E_1}$, $\widetilde{E_2}$ are smooth elliptic curves. The singular fibres are described as follows:
  \begin{itemize}
    \item The singular fibres over $\widetilde{E}_1$ are of type $IV^*$  since $(L,K,N)=(3,4,8)$ over $\widetilde{E}_1$.
    \item The singular fibres over $\widetilde{E}_2$ are of type $IV$  since $(L,K,N)=(2,2,4)$ over $\widetilde{E}_2$.
    \item The singular fibres over $\widetilde{L}$ are of type $I_1^*$  since $(L,K,N)=(2,3,7)$ over $\widetilde{L}$.
    \item The singular fibres over $\widetilde{Q}$ are of type $I_1$  since $(L,K,N)=(0,0,1)$ over $\widetilde{Q}$.
  \end{itemize}
  \item After blowing-up at $(0:0:1)\in\CC\PP^2$, the local equations of the total transforms of $G_2$ and $G_3$ are given by
  \begin{align*}
    \widehat{G_2}:
    \left\{
      \begin{array}{l}
      \displaystyle  {u'}^4{v'}^2\left( \frac{1}{12}+{u'}^3{v'}^2\varphi' \right)=0, 
      \\[15pt]
    \displaystyle  \bar{u}^2\bar{v}^4\left( \frac{1}{12}+\bar{u}\bar{v}^3\bar{\varphi} \right)=0,
      \end{array}
    \right.
  \end{align*}
  and
  \begin{align*}
    \widehat{G_3}:
    \left\{
      \begin{array}{l}
      \displaystyle  {u'}^6{v'}^3\left( \frac{1}{216}+{u'}^3{v'}^2\psi' \right)=0,
      \\[15pt]
    \displaystyle  \bar{u}^3\bar{v}^6\left( \frac{1}{216}+\bar{u}\bar{v}^3\bar{\psi} \right)=0,
      \end{array}
    \right.
  \end{align*}
  respectively, where $(u',v')$ and $(\bar{u},\bar{v})$ are the blowing-up coordinates. Furthermore, the following equations are the pull-back of the original local equation defining $W$:
  \begin{align*}
    \left\{
      \begin{array}{l}
      \displaystyle  Y^2Z=4X^3-{u'}^4{v'}^2\left( \frac{1}{12}+{u'}^3{v'}^2\varphi' \right)XZ^2- {u'}^6{v'}^3\left( \frac{1}{216}+{u'}^3{v'}^2\psi' \right)Z^3,
      \\[15pt]
    \displaystyle  Y^2Z=4X^3- \bar{u}^2\bar{v}^4\left( \frac{1}{12}+\bar{u}\bar{v}^3\bar{\varphi} \right)XZ^2- \bar{u}^3\bar{v}^6\left( \frac{1}{216}+\bar{u}\bar{v}^3\bar{\psi} \right)Z^3.
      \end{array}
    \right.
    \end{align*}
    Since these equations do not satisfy the condition (C) in Subsection \ref{Miranda's elliptic threefolds}, we change the homogeneous fibre coordinates as in the previous case (b). Then, we obtain the local equations of $\mathcal{W}$ and the discriminant as follows:
    \begin{align*}
      \left\{
        \begin{array}{l}
          \displaystyle  Y^2Z=4X^3 -{v'}^2\left( \frac{1}{12}+{u'}^3{v'}^2\varphi' \right)XZ^2 -{v'}^3\left( \frac{1}{216}+{u'}^3{v'}^2\psi' \right)Z^3,
          \\[15pt]
          \displaystyle  Y^2Z=4X^3- \bar{u}^2\left( \frac{1}{12}+\bar{u}\bar{v}^3\bar{\varphi} \right)XZ^2 -\bar{u}^3\left( \frac{1}{216}+\bar{u}\bar{v}^3\bar{\psi} \right)Z^3,
        \end{array}
      \right.
      \left\{
        \begin{array}{l}
          \widehat{\Delta}=\displaystyle {u'}^4{v'}^8\widetilde{\Delta},
          \\[15pt]
          \widehat{\Delta}=\displaystyle \bar{u}^7\bar{v}^4\widetilde{\Delta}.
        \end{array}
      \right.
    \end{align*}
    By blowing-up $\mathcal{W}$ along $E_1\cup E_2\cup \widetilde{L}$ appearing in Figure \ref{blowing up at (0:0:1)}, we have a local desingularization of $\mathcal{W}$. The singular fibres are described as follows:
    \begin{itemize}
      \item The singular fibres over generic points of $E_1$ are of type $I_2^*$ since $(L,K,N)=(2,3,8)$ over $E_1$.
      \item The singular fibres over generic points of $E_2$ are of type $I_4$ since $(L,K,N)=(0,0,4)$ over $E_2$.
      \item The singular fibres over generic points of $\widetilde{L}$ are of type $I_1^*$ since $(L,K,N)=(2,3,7)$ over $\widetilde{L}$.
      \item The singular fibres over generic points of $\widetilde{Q}$ are of type $I_1$ since $(L,K,N)=(0,0,1)$ over $\widetilde{Q}$. 
      \item The singular fibres over the collision points $q_1$, $q_2$, and $q_3$ in Figure \ref{blowing up at (0:0:1)} are of type $I_4^*$, $I_5$, and $I_3^*$, respectively from Miranda's list.
    \end{itemize}
  
  \item Since $\displaystyle\left( A_1/A_0,A_2/A_0 \right)=\left( \pm\alpha,\frac{\alpha^2}{4}\pm2 \right)$ are nodes of the discriminant locus, we do not have to blow up the base space at these two points. By blowing-up $\mathcal{W}$ over these points, we have a local desingularization of $\mathcal{W}$. The singular fibres over generic points of the discriminant locus are of type $I_1$ except for these two points. The singular fibres over these two points are of type $I_2$ on Miranda's list.
\end{enumerate}

To sum up, we have the following theorem.
\begin{theorem}
  The singular elliptic fibration $\pi_{W}\colon W\to\CC\PP^2$ is birationally equivalent to a Miranda elliptic threefold
  \begin{align*}
    \pi_{\widehat{\mathcal{W}}}\colon\widehat{\mathcal{W}}\to\widehat{\CC\PP^2},
  \end{align*}
  with the discriminant locus $\widehat{D}$ whose support is given by 
  \begin{align*}
    \mathrm{supp}\left(\widehat{D}\right)=\bigcup_{i=1}^4\left( E_{1,p_i}\cup E_{2,p_i}\cup E_{3,p_i} \right)\cup\left(\widetilde{E}_{1,(0:1:0)}\cup\widetilde{E}_{2,(0:1:0)}\right)\cup\left( E_{1,(0:0:1)}\cup E_{2,(0:0:1)} \right)\cup\widetilde{L}\cup\widetilde{Q},
  \end{align*}
  where $p_i,~(i=1,2,3,4)$ are four cusps of the original discriminant locus and $E_{1,p_i}$, $E_{2,p_i}$, and $E_{3,p_i}$ are exceptional divisors as in {\rm{Figure \ref{blowing up the cusp}}} with respect to $p_i$. The singular fibres of $\pi_{\widehat{\mathcal{W}}}$ are described as follows:
  \begin{itemize}
    \item The singular fibres over generic points of $\widetilde{Q}$ are of type $I_1$.
    \item The singular fibres over $\widetilde{L}$ are of type $I_1^*$.\item The singular fibres over generic points of $E_{1,p_i}$ are of type $II$.
    \item The singular fibres over generic points of $E_{2,p_i}$ are of type $III$.
    \item The singular fibres over generic points of $E_{3,p_i}$ are of type $I_0^*$.
    \item The singular fibres over the intersection point of $E_{1,p_i}$ with $E_{3,p_i}$, the one of $E_{2,p_i}$ with $E_{3,p_i}$, and the one of $\widetilde{Q}$ with $E_{3,p_i}$ are displayed as in {\rm{Figure \ref{The singular fibre over p_1}}}, {\rm{\ref{The singular fibre over p_2}}}, and {\rm{\ref{The singular fibre over p_3}}}, respectively.
    \item The singular fibres over $\widetilde{E}_{1,(0:1:0)}$ are of type $IV^*$.
    \item The singular fibres over $\widetilde{E}_{2,(0:1:0)}$ are of type $IV$.
    \item The singular fibres over generic points of $E_{1,(0:0:1)}$ are of type $I_2^*$.
    \item The singular fibres over generic points of $E_{2,(0:0:1)}$ are of type $I_4$.
    \item The singular fibres over the intersection point of $E_{1,(0:0:1)}$ with $E_{2,(0:0:1)}$, the one of $\widetilde{Q}$ with $E_{2,(0:0:1)}$, and the one of $\widetilde{L}$ with $E_{2,(0:0:1)}$ are of type $I_4^*$, $I_5$, and $I_3^*$, respectively.
    \item The singular fibres over the two nodes of $\widetilde{Q}$ are of type $I_2$. 
  \end{itemize}
\end{theorem}

\subsection{Monodromy of the original elliptic fibration $\pi_W$}

\indent
In this subsection, we calculate the monodromy of the original elliptic fibration $\pi_{W}\colon W\to\CC\PP^2$ constructed in Subsection \ref{Formulation of W}. By Theorem \ref{Discriminant locus}, the regular locus $\CC\PP^2\mathbin{\setminus} \mathrm{supp}\left( D \right)$ can be identified with $\CC^2\mathbin{\setminus} Q_{\mathrm{aff}}$, where $Q_{\mathrm{aff}}$ denotes the affine part of the singular quintic curve $Q$. Hence, we have the isomorphism of fundamental groups:
\begin{align*}
  \pi_1\left( \CC\PP^2\mathbin{\setminus} \mathrm{supp}\left( D \right) \right)\,\cong\pi_1\left( \CC^2\mathbin{\setminus} Q_{\mathrm{aff}} \right).
\end{align*}

We first describe the relations between the two generators around each singular point of $Q_{\mathrm{aff}}$ by Zariski--van Kampen Theorem \cite{Zariski_1929,vanKampen_1933}. Next we determine the monodromy of the elliptic fibration $\pi_{W}$.

Let $C\subset\CC^2$ be the affine plane curve which is singular at $(0,0)\in\CC^2$ defined through
\begin{align*}
  x^p-y^q=0,
\end{align*}
where $p,q\in\NN$ and $p,q\geq 2$. 
In our situation, it is enough to consider the case since the singular points of $Q_{\mathrm{aff}}$ are nodes and cusps. Let $p_1\colon\CC^2\mathbin{\setminus}C\to\CC$ be the first projection. We can take the open discs $D_{r}, D_{r'}$ in $\CC$ such that the restriction
\begin{align*}
  p\colon p_1^{-1}\left( D_{r}\setminus\{0\} \right)\cap\left( D_{r}\times D_{r'}\mathbin{\setminus}C \right)\to D_{r}\setminus\{0\},
\end{align*}
of $p_1$ is the locally trivial fibration with a section $s\colon D_{r}\setminus\{0\}\to p_1^{-1}\left( D_{r}\setminus\{0\} \right)\cap\left( D_{r}\times D_{r'}\mathbin{\setminus}C \right)$. We choose the base point $b\in D_{r}\setminus\{0\}$ such that $b$ is a real number satisfying $0<b<r$. As is well known, $\pi_1\left( D_{r}\setminus\{0\},b \right)$ acts on $\pi_1\left( p^{-1}\left( b \right),s\left( b \right) \right)$ from right, which is called the {\it{monodromy action}}. We denote this action by
\begin{align*}
  g\mapsto g^{a}~~\left(\, g\in\pi_1\left( p^{-1}\left( b \right),s\left( b \right) \right),\,a\in\pi_1\left( D_{r}\setminus\{0\},b \right) \,\right).
\end{align*}
Since the fibre $p^{-1}(b)$ is the complement to $q$ distinct points in $D_{r'}$, $\pi_1\left( p^{-1}\left( b \right),s\left( b \right) \right)$ is the free group generated by $g_1,\dots,g_q$. Moreover, $\pi_1\left( D_{r}\setminus\{0\},b \right)$ is the free group with a generator $\gamma$. Then, the following theorem holds.

\begin{theorem}[Zariski--van Kampen Theorem]\label{Zariski-vanKampen}
  Let $\iota\colon p^{-1}(b)\hookrightarrow \CC^2\mathbin{\setminus}C$ be the inclusion. Then, the homomorphism $\iota_{*}\colon\pi_1\left( p^{-1}(b),s\left( b \right) \right)\to\pi_1\left( \CC^2\mathbin{\setminus}C,s\left( b \right) \right)$ induced by $\iota$ is surjective and its kernel is the smallest normal subgroup of $\pi_1\left( p^{-1}(b),s\left( b \right) \right)$ containing the subset
  \begin{align*}
    \Set{g_i^{-1}g_i^{\gamma}|i=1,\dots,q}.
  \end{align*}
  Therefore, $\pi_1\left( \CC^2\mathbin{\setminus}C,s\left( b \right) \right)$ is isomorphic to
  \begin{align*}
    \langle g_1,\dots,g_q \left|\,g_i=g_i^{\gamma},\,i=1,\dots,q \right. \rangle.
  \end{align*}
\end{theorem}

The relation $g_i=g_i^{\gamma}$ appearing in Theorem \ref{Zariski-vanKampen} is called the {\it{monodromy relation}}.
\begin{remark}
  Theorem \ref{Zariski-vanKampen} is a special case of Zariski--van Kampen Theorem. See e.g. \cite{Zariski_1929,vanKampen_1933,Cheniot_1973} for more general settings.
\end{remark}

In the case $(p,q)=(2,2)$, $C$ has an ordinary double point at the origin. It is easy to check that $g_i^{\gamma}\,(i=1,2)$ are written as
\begin{align*}
  g_1^{\gamma}=g_2g_1g_2^{-1},\,g_2^{\gamma}=g_2g_1g_2g_1^{-1}g_2^{-1}.
\end{align*}
Hence, the monodromy relations amount to
\begin{align*}
  g_1g_2=g_2g_1.
\end{align*}
In the case $(p,q)=(3,2)$, $C$ has a cusp at the origin. In this case, we can verify that $g_i^{\gamma}\,(i=1,2)$ are written as
\begin{align*}
  g_1^{\gamma}=g_2g_1g_2g_1^{-1}g_2^{-1},\,g_2^{\gamma}=g_2g_1g_2g_1g_2^{-1}g_1^{-1}g_2^{-1}.
\end{align*}
Thus, the monodromy relations yield the following one:
\begin{align*}
  g_1g_2g_1=g_2g_1g_2.
\end{align*}
Therefore, around a node of $Q_{\mathrm{aff}}$, the generators $a_1,a_2$ corresponding to the closed arcs in $\CC^2\setminus Q_{\mathrm{aff}}$ described as in Figure \ref{figure_node} satisfy
\begin{align}
  a_1b_1=b_1a_1.\label{relation_node}
\end{align}
On the other hand, around a cusp of $Q_{\mathrm{aff}}$, the generators $a_2,b_2$ corresponding to the closed arcs in $\CC^2\setminus Q_{\mathrm{aff}}$ described as in Figure \ref{figure_cusp} satisfy\begin{align}
  a_2b_2a_2=b_2a_2b_2.\label{relation_cusp}
\end{align} 
\begin{figure}[H]
  \centering
  \begin{minipage}{0.4\columnwidth}
    \centering
    \begin{tikzpicture}
      \draw (-2,1) -- (2,-1);
      \draw (-2,-1) -- (2,1);
      \draw[<-] (-1,0.6) arc (20:300:0.15);
      \draw[<-] (-1,-0.4) arc (60:360:0.15);
      \fill[black] (-1.3,0.9) circle(0) node{$a_1$};
      \fill[black] (-1.3,-0.1) circle(0) node{$b_1$};
    \end{tikzpicture}
    \caption{Node}
    \label{figure_node}
  \end{minipage}
  \begin{minipage}{0.4\columnwidth}
    \centering
    \begin{tikzpicture}
      \draw (-2,1) to [out=330,in=110] (0,-1);
      \draw (2,1) to [out=210,in=70] (0,-1);
      \draw[<-] (-1,0.5) arc (30:300:0.15);
      \draw[<-] (1.1,0.3) arc (-100:180:0.15);
      \fill[black] (-1.3,0.8) circle(0) node{$a_2$};
      \fill[black] (1.2,0.8) circle(0) node{$b_2$};
    \end{tikzpicture}
    \caption{Cusp}
    \label{figure_cusp}
  \end{minipage}
\end{figure}

Next we consider the monodromy of the elliptic fibration $\pi_{W}$. Let $p$ be the reference point of $\pi_1\left( \CC^2\mathbin{\setminus}Q_{\mathrm{aff}} \right)$. The monodromy of $\pi_W$ is the representation 
\begin{align*}
  \rho\colon\pi_1\left( \CC^2\mathbin{\setminus}Q_{\mathrm{aff}} \right)\to\mathrm{Aut}\left( H_1\left( \pi_W^{-1}(p),\ZZ \right) \right),
\end{align*}
where the fibre $\pi_W^{-1}\left( p \right)$ is a regular fibre, that is, $\pi_W^{-1}\left( p \right)$ is an elliptic curve. It is known that the automorphism group of the first homology $\mathrm{Aut}\left( H_1\left( \pi_W^{-1}(p),\ZZ \right) \right)$ is isomorphic to $SL\left( 2,\ZZ \right)$. Thus, we can regard the monodromy representation as
\begin{align*}
  \rho\colon\pi_1\left( \CC^2\mathbin{\setminus}Q_{\mathrm{aff}} \right)\to SL\left( 2,\ZZ \right).
\end{align*}

The monodromy matrices of all the types of singular fibres in elliptic surfaces were found by Kodaira \cite{KodairaI-III} up to $SL\left( 2,\ZZ \right)$-conjugacy. Let $\mathrm{Reg}\left( Q_{\mathrm{aff}} \right)$ be the set of smooth points of $Q_{\mathrm{aff}}$. Noting that the singular fibres over the points of $\mathrm{Reg}\left(Q_{\mathrm{aff}} \right)$ are of type $I_1$ in Kodaira's notation and the monodromy matrix of type $I_1$ is given by $\displaystyle\begin{pmatrix} 1 & 1 \\ 0 & 1 \\ \end{pmatrix}$ up to conjugacy in $SL\left( 2,\ZZ \right)$, the following theorem holds.

\begin{theorem}
  With respect to a suitable choice of basis for $H_1\left( \pi_{W}^{-1}\left( p \right),\ZZ \right)$ where $p$ is the reference point of $\pi_1\left( \CC^2\mathbin{\setminus} Q_{\mathrm{aff}} \right)$, the monodromy representation of the original elliptic fibration $\pi_W$ is characterized as follows:
  \begin{itemize}
    \item For each homotopy class of the closed arcs in $\CC^2\setminus Q_\mathrm{{aff}}$ described as in Figure $\ref{figure_node}$ and Figure $\ref{figure_cusp}$, the corresponding monodromy matrix is either 
    \begin{align*}
      \begin{pmatrix}
        1 & 1 \\
        0 & 1 \\
      \end{pmatrix}
      \,\text{or}\,
      \begin{pmatrix*}[r]
        1 & 0 \\
        -1 & 1 \\
      \end{pmatrix*}.
    \end{align*}
    \item Around a node of $Q_\mathrm{{aff}}$, the above monodromy matrices are the same.
    \item Around a cusp of $Q_\mathrm{{aff}}$, the above monodromy matrices are distinct.
  \end{itemize}
\end{theorem}

\begin{proof}
  Let $A_1,\, B_1,\, A_2,\, B_2\in SL\left( 2,\ZZ \right)$ denote the monodromy matrices corresponding to the generators $a_1,\, b_1,\, a_2,\, b_2$ as in Figure \ref{figure_node} and Figure \ref{figure_cusp}, respectively. 
  
  Since the generators $a_2,b_2$ satisfy the relation \eqref{relation_cusp}, the monodromy matrices $A_2, B_2$ satisfy the relation
  \begin{align*}
    A_2B_2A_2=B_2A_2B_2.
  \end{align*}
  Without loss of generality, we can assume that $A_2=\begin{pmatrix} 1 & 1 \\ 0 & 1 \\ \end{pmatrix}$. Noting that $B_2$ is conjugate to $\begin{pmatrix} 1 & 1 \\ 0 & 1 \\ \end{pmatrix}$ in $SL\left( 2,\ZZ \right)$, we obtain
  \begin{align*}
    B_2=
    \begin{pmatrix*}[r]
      1 & 0 \\
      -1 & 1 \\
    \end{pmatrix*}.
  \end{align*}

  On the other hand, since the generators $a_1,b_1$ satisfy the relation \eqref{relation_node}, the monodromy matrices $A_1, B_1$ satisfy the relation
  \begin{align*}
    A_1B_1=B_1A_1.
  \end{align*}
  From the above discussion, we should consider that $A_1$ is given as following two matrices:
  \begin{align*}
    \begin{pmatrix}
      1 & 1 \\
      0 & 1 \\
    \end{pmatrix}
    ,\,
    \begin{pmatrix*}[r]
      1 & 0 \\
      -1 & 1 \\
    \end{pmatrix*}.
  \end{align*}
  In the case $A_1=\begin{pmatrix}1 & 1 \\0 & 1 \\\end{pmatrix}$, we have
  \begin{align*}
    B_1=
    \begin{pmatrix}
      1 & 1 \\
      0 & 1 \\
    \end{pmatrix},
  \end{align*}
  since $B_1$ is conjugate to $\begin{pmatrix}1 & 1 \\0 & 1 \\\end{pmatrix}$ in $SL\left( 2,\ZZ \right)$. In the other case, noting that the matrix $B_1$ is conjugate to $\begin{pmatrix*}[r] 1 & 0 \\ -1 & 1 \\ \end{pmatrix*}$ in $SL\left( 2,\ZZ \right)$, we obtain
  \begin{align*}
    B_1=
    \begin{pmatrix*}[r]
      1 & 0 \\
      -1 & 1 \\
    \end{pmatrix*}.
  \end{align*}
\end{proof}

\section*{Acknowledgments}
The author would like to thank Daisuke Tarama (Ritsumeikan University) for helpful comments and discussions. This work is supported by JST SPRING Grant Number JPMJSP2101.




\bibliographystyle{aomalpha}
\bibliography{refs_AG,refs_DS,text}


\medskip
\begin{flushleft}
  Genki Ishikawa
  \\
  Department of Mathematical Sciences, Ritsumeikan University.
  \\
  1-1-1 Nojihigashi, Kusatsu, Shiga, 525-8577, Japan.
  \\
  E-mail address: \texttt{ra0070fs@ed.ritsumei.ac.jp}
\end{flushleft}
\end{document}